\documentclass[11pt]{article}
\usepackage{algorithm2e}
\usepackage{framed}
\usepackage{amssymb}
\usepackage{amsfonts}
\usepackage{amsmath}
\usepackage{graphicx}
\usepackage{subfigure}
\usepackage{url}
\usepackage{tabularx}

\newtheorem{theorem}{Theorem}

\newtheorem{lemma}{Lemma}

\newtheorem{claim}{Claim}

\newenvironment{proof}{\noindent {\em Proof:}}{\\\hspace*{\fill}\mbox{$\diamond$}}

\usepackage{nicefrac}
\newcommand{\nfrac}{\nicefrac}

\newcommand{\sdp}{{\sf SDP}}

\newcommand{\defeq}{\stackrel{\textup{def}}{=}}

\newcommand{\vol}{\mathrm{vol}}

\begin{document}

%-----------------------------------------------------------------------
%%%
%\begin{titlepage}

%\title{A Spectral Algorithm for Improving Graph Partitions with Applications to Exploring Data Graphs Locally}
\title{A Local Spectral Method for Graphs: \\ 
with Applications to Improving Graph Partitions \\
and Exploring Data Graphs Locally}

\author{
Michael W. Mahoney
\thanks{
Department of Mathematics,
Stanford University,
Stanford, CA 94305.
{\tt mmahoney@cs.stanford.edu}.
}
\and
Lorenzo Orecchia
\thanks{
Computer Science Division, 
UC Berkeley,
Berkeley, CA, 94720.
{\tt orecchia@eecs.berkeley.edu}.
}
\and 
Nisheeth K. Vishnoi
\thanks{
Microsoft Research,
Bangalore, India.
{\tt nisheeth.vishnoi@gmail.com}.
}
}

\date{}
\maketitle

%\begin{center}
%{\bf
%\emph{
%NOTE:  Draft version from \today.  \\
%XXX.  PLEASE DO NOT DISTRIBUTE THIS VERSION.
%}
%}
%\end{center}

\vspace{3mm}
\begin{abstract}
The second eigenvalue of the Laplacian matrix and its associated eigenvector 
are fundamental features of an undirected graph, and as such they have found 
widespread use in scientific computing, machine learning, and data analysis.  
In many applications, however, graphs that arise have several \emph{local} 
regions of interest, and the second eigenvector will typically fail to 
provide information fine-tuned to each local region.  
In this paper, we introduce a locally-biased analogue of the second 
eigenvector, and we demonstrate its usefulness at highlighting local 
properties of data graphs in a semi-supervised manner.
To do so, we first view the second eigenvector as the solution to a 
constrained optimization problem, and we incorporate the local information 
as an additional constraint; 
we then characterize the optimal solution to this new problem and show that 
it can be interpreted as a generalization of a Personalized PageRank vector; 
and finally, as a consequence, we show that the solution can be computed in 
nearly-linear time. 
In addition, we show that this locally-biased vector can be used to compute 
an approximation to the best partition \emph{near} an input seed set in a 
manner analogous to the way in which the second eigenvector of the Laplacian 
can be used to obtain an approximation to the best partition in the entire 
input graph.
Such a primitive is useful for identifying and refining clusters locally, as 
it allows us to focus on a local region of interest in a semi-supervised 
manner.
Finally, we provide a detailed empirical evaluation of our method by showing 
how it can applied to finding locally-biased sparse cuts around an input 
vertex seed set in social and information networks.
\end{abstract}
\vspace{3mm}

\section{Introduction}
\label{sxn:intro}

Spectral methods are popular in machine learning, data analysis, and applied 
mathematics due to  their strong underlying theory and their good 
performance in a wide range of applications.
In the study of undirected graphs, in particular, spectral techniques play 
an important role, as many fundamental structural properties of a graph 
depend directly on spectral quantities associated with matrices representing 
the graph.
Two fundamental objects of study in this area are the second smallest 
eigenvalue of the graph Laplacian and its associated eigenvector.
These quantities determine many features of the graph, including the 
behavior of random walks and the presence of sparse cuts.
This relationship between the graph structure and an easily-computable 
quantity has been exploited in data clustering, community detection, image 
segmentation, parallel computing, and many other applications.

A potential drawback of using the second eigenvalue and its associated 
eigenvector is that they are inherently \emph{global} quantities, and thus 
they may not be sensitive to very \emph{local} information. 
For instance, a sparse cut in a graph may be poorly correlated with the 
second eigenvector (and even with all the eigenvectors of the Laplacian) 
and thus invisible to a method based only on eigenvector analysis. 
Similarly, based on domain knowledge one might have information about a 
specific target region in the graph, in which case one might be interested 
in finding clusters only near this prespecified local region, \emph{e.g.}, 
in a semi-supervised manner; but this local region might be essentially 
invisible to a method that uses only global eigenvectors.
For these and related reasons, standard global spectral techniques can have 
substantial difficulties in semi-supervised settings, where the goal is to 
learn more about a locally-biased target region of the graph.

In this paper, we provide a methodology to construct a locally-biased 
analogue of the second eigenvalue and its associated eigenvector, and we 
demonstrate both theoretically and empirically that this localized vector 
inherits many of the good properties of the global second eigenvector. 
Our approach is inspired by viewing the second eigenvector as the optimum of 
a constrained global quadratic optimization program. 
To model the localization step, we modify this program by adding a natural 
locality constraint.
This locality constraint requires that any feasible solution have sufficient 
correlation with the target region, which we assume is given as input in 
the form of a set of nodes or a distribution over vertices. 
The resulting optimization problem, which we name \textsf{LocalSpectral} and 
which is displayed in Figure~\ref{fig:spectral},
% and discussed in detail in Section~\ref{sxn:optimize}, 
is the main object of our work. 

The main advantage of our formulation is that an optimal solution to 
\textsf{LocalSpectral} captures many of the same structural properties as 
the global eigenvector, except in a locally-biased setting. 
For example, as with the global optimization program, our locally-biased 
optimization program has an intuitive geometric interpretation.
Similarly, as with the global eigenvector, an optimal solution to 
\textsf{LocalSpectral} is efficiently computable. 
To show this, we characterize the optimal solutions of 
\textsf{LocalSpectral} and show that such a solution can be constructed in 
nearly-linear time by solving a system of linear equations.
In applications where the eigenvectors of the graph are pre-computed and 
only a small number of them are needed to describe the data, the optimal 
solution to our program can be obtained by performing a small number of 
inner product computations. 
Finally, the optimal solution to \textsf{LocalSpectral} can be used to 
derive bounds on the mixing time of random walks that start near the local 
target region as well as on the existence of sparse cuts near the 
locally-biased target region. 
In particular, it lower bounds the conductance of cuts as a function of how 
well-correlated they are with the seed vector.
This will allow us to exploit the analogy between global eigenvectors and 
our localized analogue to design an algorithm for discovering sparse cuts 
near an input seed set of vertices. 

In order to illustrate the empirical behavior of our method, we will 
describe its performance on the problem of finding locally-biased sparse 
cuts in real data graphs.
Subsequent to the dissemination of the initial technical report version of 
this paper, our methodology was applied to the problem of finding, given a 
small number of ``ground truth'' labels that correspond to known segments 
in an image, the segments in which those labels reside~\cite{MVM11}.
This computer vision application will be discussed briefly.
Then, we will describe in detail how our algorithm for discovering sparse 
cuts near an input seed set of vertices may be applied to the problem of 
exploring data graphs locally and to identifying locally-biased clusters and 
communities in a more difficult-to-visualize social network application.
In addition to illustrating the performance of the method in a practical 
application related to the one that initially motivated 
this work~\cite{LLDM08_communities_CONF,LLDM09_communities_IM,LLM10_communities_CONF}, 
this social graph application will illustrate how the various ``knobs'' of our method can 
be used in practice to explore the structure of data graphs in a locally-biased 
manner.

Recent theoretical work has focused on using spectral ideas to find good 
clusters nearby an input seed set of 
nodes~\cite{Spielman:2004,andersen06local,chung07_fourproofs}. 
These methods are based on running a number of local random walks around the 
seed set and using the resulting distributions to extract information about 
clusters in the graph.
Recent empirical work has used Personalized PageRank, a particular variant 
of a local random walk, to characterize very finely the clustering and 
community structure in a wide range of very large social and information 
networks~\cite{andersen06seed,LLDM08_communities_CONF,LLDM09_communities_IM,LLM10_communities_CONF}. 
In contrast with previous methods, our local spectral method is the first to 
be derived in a direct way from an explicit optimization problem inspired 
by the global spectral problem.  
Interestingly, our characterization also shows that optimal solutions to 
\textsf{LocalSpectral} are  generalizations of Personalized PageRank, 
providing an additional insight to why local random walk methods work well 
in practice.

%\paragraph{Organization of the paper.} 
In the next section, we will describe relevant background and notation; and 
then, in Section~\ref{sxn:optimize}, we will present our formulation of a 
locally-biased spectral optimization program, the solution of which will 
provide a locally-biased analogue of the second eigenvector of the graph 
Laplacian. 
Then, in Section~\ref{sxn:partition} we will describe how our method may be 
applied to identifying and refining locally-biased partitions in a graph; and
in Section~\ref{sxn:empirical} we will provide a detailed empirical 
evaluation of our algorithm.
Finally, in Section~\ref{sxn:discussion}, we will conclude with a discussion 
of our results in a broader~context.

\section{Background and Notation.} 
\label{sxn:background}

Let $G=(V,E,w)$ be a connected undirected graph with $n=|V|$ vertices and 
$m=|E|$ edges, in which edge $\{i,j\}$ has weight $w_{ij}.$
For a set of vertices $S \subseteq V$ in a graph, the \emph{volume of $S$} 
is $\vol(S) \defeq \sum_{i \in S}d_i$, in which case the \emph{volume of the
graph $G$} is $\vol(G) \defeq \vol(V)=2m$. 
In the following, $A_G \in \mathbb{R}^{V \times V}$ will denote the 
adjacency matrix of $G$, while $D_G \in \mathbb{R}^{V \times V}$ will denote 
the diagonal degree matrix of $G$, \emph{i.e.},  
$D_G(i,i)=d_i = \sum_{\{i,j\} \in E} w_{ij}$, the weighted degree of vertex $i$.
The Laplacian of $G$ is defined as $L_G \defeq D_G-A_G$. 
(This is also called the combinatorial Laplacian, in which case the
normalized Laplacian of $G$ is $\mathcal{L}_G\defeq D_G^{-1/2}L_GD_G^{-1/2}$.)

The Laplacian is the symmetric matrix having quadratic form 
$x^T L_G x = \sum_{ij \in E} w_{ij} (x_i - x_j)^2$, for $x \in  \mathbb{R}^V$. 
This implies that $L_G$ is positive semidefinite and that the all-one vector 
$1 \in  \mathbb{R}^V$ is the eigenvector corresponding to the smallest 
eigenvalue $0$.
For a symmetric matrix $A$, we will use $A \succeq 0$ to denote that it is 
positive semi-definite.
Moreover, given two symmetric matrices $A$ and $B$, the expression 
$A \succeq B$ will mean $A - B \succeq 0$.
Further, for two $n\times n$ matrices $A$ and $B$, we let $A \circ B$ 
denote ${\rm Tr}\:(A^TB)$. 
Finally, for a matrix $A,$ let $A^{+}$ denote its (uniquely defined) 
Moore-Penrose pseudoinverse.

For two vectors $x,y \in \mathbb{R}^{n}$, and the degree matrix $D_G$ for a 
graph $G$, we define the degree-weighted inner product as
$ x^T D_G y \defeq  \sum_{i=1}^{n} x_{i}y_{i}d_{i}$.
Given a subset of vertices $S \subseteq V$, we denote by $1_S$ the indicator 
vector of $S$ in $\mathbb{R}^V$ and by $1$ the vector in $\mathbb{R}^V$ 
having all entries set equal to $1$.
We consider the following definition of the complete graph $K_n$ on the 
vertex set $V$: $ A_{K_n} \defeq  \frac{1}{\vol(G)} D_G 11^T D_G$.
Note that this is not the standard complete graph, but a weighted version of 
it, where the weights depend on $D_G$.
With this scaling we have $D_{K_n} = D_G$.
Hence, the Laplacian of the complete graph defined in this manner becomes
$ L_{K_n} = D_G - \frac{1}{\vol(G)} D_G 11^T D_G $.

In this paper, the \emph{conductance $\phi(S)$ of a cut $(S,\bar{S})$} is 
$\phi(S) \defeq \vol(G) \cdot \frac{|E(S,\bar{S})|}{\vol(S)\cdot\vol(\bar{S})}$.
A sparse cut, also called a good-conductance partition, is one for which 
$\phi(S)$ is small.
The \emph{conductance of the graph $G$} is then
$\phi(G)=\min_{S \subseteq V} \phi(S)$.
Note that the conductance of a set $S$, or equivalently a cut $(S,\bar{S})$, is often 
defined as $\phi^\prime(S)=|E(S,\bar{S})|/\min\{\vol(S),\vol(\bar{S})\}$. 
This notion is equivalent to that $\phi(S)$, in that the value 
$\phi(G)$ thereby obtained for the conductance of the graph $G$ differs by 
no more than a factor of $2$ times the constant $\vol(G)$, depending on 
which notion we use for the conductance of a~set.  
%Note also that, up to that constant, the conductance of a graph is the same 
%as its normalized cut value.
%Note that $\phi(G)$ is defined so that it is a \emph{dimension-less} quantity, in the sense that it does not change if each edge in the graph is replaced by the same number of multiple~edges. 

\section{The \textsf{LocalSpectral} Optimization Program}
\label{sxn:optimize}

In this section, we introduce the local spectral optimization program 
$\textsf{LocalSpectral}(G, s, \kappa)$ as a strengthening of the usual
global spectral program $\mathsf{Spectral}(G)$.
% written in 
% Figure~\ref{fig:spectral} as an optimization problem over vectors in 
% $\mathbb{R}^V$. 
To do so, we will augment $\mathsf{Spectral}(G)$ with a locality constraint
of the form $(x^TD_Gs)^2\geq\kappa$, for a seed vector $s$ and a correlation
parameter $\kappa$.
Both these programs are homogeneous quadratic programs, with optimization 
variable the vector $x \in \mathbb{R}^V$, and thus any solution vector $x$ 
is essentially equivalent to $-x$ for the purpose of these optimizations. 
Hence, in the following we do not differentiate between $x$ and $-x$, and we 
assume a suitable direction is chosen in each instance.

\subsection{Motivation for the Program}
\label{sxn:optimize-program}

Recall that the second eigenvalue $\lambda_2(G)$ of the Laplacian $L_G$ can 
be viewed as the optimum of the standard optimization problem 
$\mathsf{Spectral}(G)$ described in Figure~\ref{fig:spectral}. 
In matrix terminology, the corresponding optimal solution $v_2$ is a 
generalized eigenvector of $L_G$ with respect to $D_G$. 
For our purposes, however, it is best to consider the geometric meaning of 
this optimization formulation. 
To do so, suppose we are operating in a vector space $ \mathbb{R}^V$, where 
the $i$th dimension is stretched by a factor of $d_i,$ so that the natural 
identity operator is $D_G$ and the inner product between two vectors $x$ and 
$y$ is given by $\sum_{i \in V} d_i x_i y_i = x^T D_{G} y$.
In this representation, $\mathsf{Spectral}(G)$ is seeking the vector 
$x \in  \mathbb{R}^V$ that is orthogonal to the all-one vector, lies on the 
unit sphere, and minimizes the Laplacian quadratic form. 
Note that such an optimum $v_2$ may lie anywhere on the unit sphere.  

\begin{figure}
\begin{minipage}{0.5\textwidth}
\begin{alignat*}{4}
  &\text{min} & x^T  L_{G} x \\
                     &\text{s.t.} & x^T D_{G} x = 1  \\
                     &            & (x^T D_{G} 1)^2 = 0  \\
		  &            & x \in \mathbb{R}^V
\end{alignat*}
\end{minipage}
\begin{minipage}{0.5\textwidth}
\begin{alignat*}{4}
 &\text{min} & x^T  L_{G} x               \\
                     &\text{s.t.} & x^T  D_{G} x = 1           \\
		   &            & (x^T D_{G} 1)^2 = 0  \\
                     &            &(x^T D_{G} s) ^2 \geq \kappa   \\
		  &            & x \in \mathbb{R}^V
\end{alignat*}
\end{minipage}
\caption{Global and local spectral optimization programs.
Left: The usual spectral program $\mathsf{Spectral}(G)$. 
Right: Our new locally-biased spectral program 
\textsf{LocalSpectral}$(G,s,\kappa)$.  
In both cases, the optimization variable is the vector $x \in \mathbb{R}^{n}$.
%XXX.  SAY $\mathbb{R}^{n}$ AND $\mathbb{R}^{V}$ ARE THE SAME.
}
\label{fig:spectral}
\end{figure}

% \begin{figure}
% \begin{minipage}{0.5\textwidth}
% \begin{alignat*}{4}
% %\label{prog:spectral-global1}
%   &\text{minimize} & x^T  L_{G} x \\
% \nonumber
%                      &\text{s.t.} & \langle x,x \rangle_{D} = 1  \\
% \nonumber
%                      &            & \langle x,1\rangle _{D} = 0  
% \end{alignat*}
% \end{minipage}
% \begin{minipage}{0.5\textwidth}
% \begin{alignat*}{4}
% %\label{prog:spectral-local-p1A}
%  &\text{minimize} & x^T  L_{G} x               \\
% %\label{prog:spectral-local-p1B}
%                      &\text{s.t.} & x^T  L_{K_n} x = 1           \\
% %\label{prog:spectral-local-p1C}
%                      &            &\langle x,s \rangle_{D}^2 \geq \kappa   
% \end{alignat*}
% \end{minipage}
% \caption{
% Left: Spectral relaxation $\mathsf{Spectral}(G)$. 
% Right: Our relaxation \textsf{LocalSpectral}$(G,s,\kappa)$.  
% In both cases, the optimization variable is $x \in \mathbb{R}^{n}$.
% XXX.  TO BE REMOVED EVENTUALLY.
% }
% \label{fig:spectral}
% \end{figure}

Our goal here is to modify $\mathsf{Spectral}(G)$ to incorporate a bias 
towards a target region which we assume is given to us as an input vector $s$.
We will assume (without loss of generality) that $s$ is properly normalized 
and orthogonalized so that $s^T D_{G} s =1$ and $s^T D_{G} 1 =0$. 
While $s$ can be a general unit vector orthogonal to $1$, it may be helpful 
to think of $s$ as the indicator vector of one or more vertices in $V$, 
corresponding to the target region of the graph.
We obtain $\textsf{LocalSpectral}(G,s,\kappa)$ from $\mathsf{Spectral}(G)$ 
by requiring that a feasible solution also have a sufficiently large 
correlation with the vector $s$. 
This is achieved by the addition of the constraint 
$(x^T D_{G} s)^2 \geq \kappa$, which ensures that the projection of $x$ onto 
the direction $s$ is at least $\sqrt{\kappa}$ in absolute value, where the 
parameter $\kappa$ is also an input parameter ranging between $0$ and $1$. 
Thus, we would like the solution to be well-connected with or to lie near 
the \emph{seed} vector $s$. 
In particular, as displayed pictorially in Figure~\ref{fig:sphere}, $x$ must 
lie within the spherical cap centered at $s$ that contains all vectors at an 
angle of at most $\arccos(\sqrt{\kappa})$ from $s$. 
Thus, higher values of $\kappa$ demand a higher correlation with $s$ and, 
hence, a stronger localization. 
Note that in the limit $\kappa = 0$, the spherical cap constituting the 
feasible region of the program is guaranteed to include $v_2$ and 
$\textsf{LocalSpectral}(G,s,\kappa)$ is equivalent to $\mathsf{Spectral}(G)$. 
In the rest of this paper, we refer to $s$ as the \emph{seed vector} and to 
$\kappa$ as the \emph{correlation parameter} for a given 
$\textsf{LocalSpectral}(G,s,\kappa)$ optimization problem.  
Moreover, we denote the objective value of the program
$\textsf{LocalSpectral}(G,s,\kappa)$ by the number $\lambda(G,s,\kappa)$. 

\begin{figure}[h]
%\begin{figure}[t]
   \begin{center}
   \includegraphics[ scale=.17]{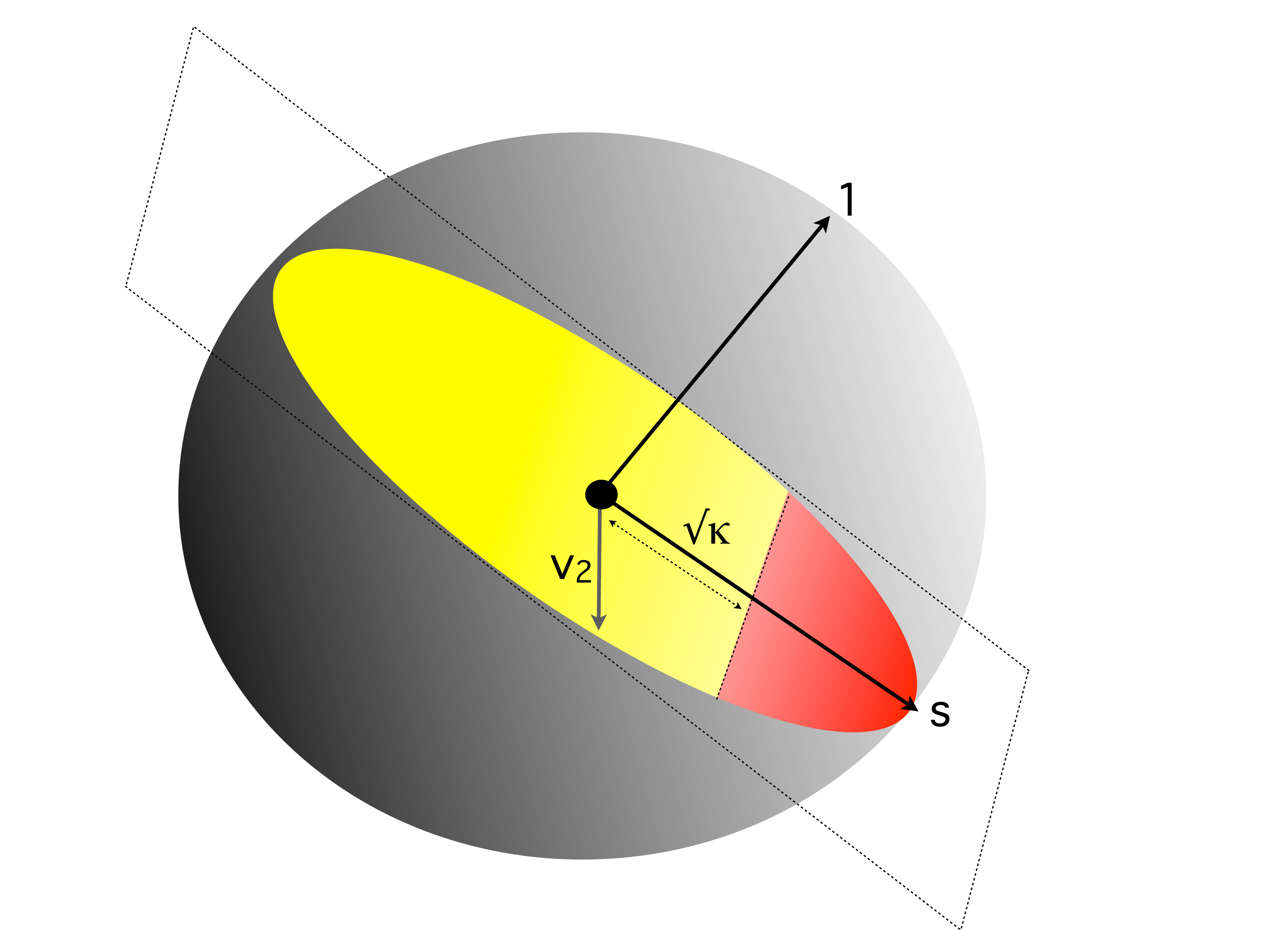}
   \end{center}
\caption{(Best seen in color.) Pictorial representation of the feasible 
regions of the optimization programs $\mathsf{Spectral}(G)$ and 
$\textsf{LocalSpectral}(G,s,\kappa)$ that are defined in 
Figure~\ref{fig:spectral}.  See the text for a discussion.}
\label{fig:sphere}	
\end{figure}

\subsection{Characterization of the Optimal Solutions of \textsf{LocalSpectral}}
\label{sxn:optimize-theory}

Our first theorem is a characterization of the optimal solutions of 
\textsf{LocalSpectral}. 
Although \textsf{LocalSpectral} is a non-convex program (as, of course, is 
\textsf{Spectral}), the following theorem states that solutions to it can be 
expressed as the solution to a system of linear equations which has a 
natural interpretation.
The proof of this theorem (which may be found in 
Section~\ref{sxn:optimize-proofs-pagerank}) will involve a relaxation of the 
non-convex program \textsf{LocalSpectral} to a convex semidefinite program (SDP), 
\emph{i.e.}, the variables in the optimization program will be distributions 
over vectors rather than the vectors themselves.
For the statement of this theorem, recall that $A^{+}$ denotes the (uniquely 
defined) Moore-Penrose pseudoinverse of the matrix $A$.

\begin{theorem}[Solution Characterization]
\label{thm:pagerank}
Let $s \in \mathbb{R}^{V}$ be a seed vector such that $s^T D_G 1 =0$, 
$s^T D_G s = 1$, and $s^T D_G v_2 \neq 0$, where $v_{2}$ is the second 
generalized eigenvector of $L_G$ with respect to $D_G$.
In addition, let $1> \kappa \geq 0$ be a correlation parameter, and let
$x^{\star}$ be an optimal solution to \textsf{LocalSpectral}$(G,s,\kappa)$.
Then, there exists some $\gamma \in (-\infty, \lambda_{2}(G))$ and a 
$c \in [0, \infty]$ such that 
\begin{equation}
\label{eqn:xstar}
 x^{\star} = c(L_{G}-\gamma D_G)^{+} D_G s. 
\end{equation}
\end{theorem}

\noindent
There are several parameters (such as $s$, $\kappa$, $\gamma$, and $c$) in 
the statement of Theorem~\ref{thm:pagerank}, and understanding their 
relationship is important:
$s$ and $\kappa$ are the parameters of the program; $c$ is a normalization 
factor that rescales the norm of the solution vector to be $1$ (and that can 
be computed in linear time, given the solution vector); and $\gamma$ is 
implicitly defined by $\kappa$, $G$, and~$s$. 
The correct setting of $\gamma$ ensures that 
$(s^T D_{G} x^\star)^2 = \kappa,$ \emph{i.e.}, that $x^\star$ is found 
exactly on the boundary of the feasible region.
At this point, it is important to notice the behavior of $x^\star$ and 
$\gamma$ as $\kappa$ changes. 
As $\kappa$ goes to $1$, $\gamma$ tends to $-\infty$ and $x^\star$ 
approaches $s$;
conversely, as $\kappa$ goes to $0$,  $\gamma$ goes to $\lambda_2(G)$ and 
$x^\star$ tends towards $v_2$, the global eigenvector.
We will discuss how to compute $\gamma$ and $x^\star$, given a specific 
$\kappa$, in Section~\ref{sxn:optimize-comp}.

Finally, we should note that there is a close connection between the 
solution vector $x^\star$ and the popular PageRank procedure. 
Recall that PageRank refers to a method to determine a global rank or global 
notion of importance for a node in a graph such as the web that is based on 
the link structure of the graph~\cite{BP98,LM04,berkhin05_pagerank}. 
There have been several extensions to the basic PageRank concept, including 
Topic-Sensitive PageRank~\cite{haveliwala03_topicpr}
and Personalized PageRank~\cite{JW03}. 
In the same way that PageRank can be viewed as a way to express the quality 
of a web page over the entire web, Personalized PageRank expresses a 
link-based measure of page quality around user-selected pages.
In particular, given a vector $s \in \mathbb{R}^{V}$ and a 
\emph{teleportation} constant $\alpha> 0$, the Personalized PageRank vector 
can be written as
$\mbox{pr}_{\alpha,s}=\left(L_{G}+\frac{1-\alpha}{\alpha}D_{G}\right)^{-1}D_{G}s
$~\cite{andersen06local}.
By setting $\gamma = -\frac{1-\alpha}{\alpha},$  the optimal solution to \textsf{LocalSpectral} is proved to 
be a generalization of Personalized PageRank. 
In particular, this means that for high values of the correlation parameter 
$\kappa$, for which the corresponding $\gamma$ in Theorem~\ref{thm:pagerank}
is negative, the optimal solution to \textsf{LocalSpectral} takes the form 
of a Personalized PageRank vector. 
On the other hand, when $\gamma \geq 0,$ the optimal solution to 
\textsf{LocalSpectral} provides a smooth way of transitioning from the 
Personalized PageRank vector to the global second eigenvector~$v_2$.

\subsection{Computation of the Optimal Solutions of \textsf{LocalSpectral}}
\label{sxn:optimize-comp}

In this section, we discuss how to compute efficiently an optimal solution
for $\textsf{LocalSpectral}(G,s, \kappa)$, for a fixed choice of the 
parameters $G$, $s$, and $\kappa$. 
The following theorem is our main result.

\begin{theorem}[Solution Computation]
\label{thm:comp}
For any $\varepsilon >0$, a solution to \textsf{LocalSpectral}$(G,s,\kappa)$ 
of value at most $(1+\varepsilon )\cdot \lambda (G,s,\kappa)$ can be 
computed in time 
$\tilde{O}(\nfrac{m}{\sqrt{\lambda_{2}(G)}} \cdot \log(\nfrac{1}{\varepsilon }))$ 
using the Conjugate Gradient Method~\cite{GVL96}. 
Alternatively, such a solution can be computed in time 
$\tilde{O}(m \log(\nfrac{1}{\varepsilon }))$ using the Spielman-Teng 
linear-equation solver~\cite{Spielman:2004}. 
\end{theorem}
\begin{proof}
By Theorem~\ref{thm:pagerank}, we know that the optimal solution $x^\star$ must be a unit-scaled version of 
$y(\gamma) =(L_{G}-\gamma D_G)^{+} D_G s,$ for an appropriate choice of $\gamma \in (-\infty, \lambda_{2}(G)).$
Notice that, given a fixed $\gamma,$ the task of computing $y(\gamma)$ is equivalent to solving the system of linear equations
$(L_{G}-\gamma D_G) y = D_{G}s$ for the unknown $y.$ This operation can be performed, up to accuracy $\varepsilon ,$ in time $\tilde{O}(\nfrac{m}{\sqrt{\lambda_{2}(G)}} \cdot \log(\nfrac{1}{\varepsilon }))$ 
using the Conjugate Gradient Method, or  in time 
$\tilde{O}(m \log(\nfrac{1}{\varepsilon }))$ using the Spielman-Teng linear-equation 
solver.
To find the correct setting of $\gamma,$ it suffices to perform a binary search over the possible values of $\gamma$ in the interval $(-\vol(G), \lambda_2(G)),$ until $(s^T D_{G} x)^2$ is sufficiently close to $\kappa.$ 
\end{proof}

\noindent
We should note that, depending on the application, other methods of 
computing a solution to $\textsf{LocalSpectral}(G,s,\kappa)$ might be more 
appropriate.
In particular, if an eigenvector decomposition of $L_G$ has been 
pre-computed, as is the case in certain machine learning and data analysis 
applications, then this computation can be modified as follows.
Given an eigenvector decomposition of $L_G$ as 
$ L_G = \sum_{i=2}^n \lambda_i D_G^{1/2} u_i u_i^T D_G^{1/2} $, then
$y(\gamma)$ must take the form
$$
y(\gamma) = (L_{G}-\gamma D_G)^{+} D_G s= \sum_{i=2}^n \frac{1}{\lambda_i - \gamma} (s^T D_G^{1/2} u)^2  ,
$$
for the same choice of $c$ and $\gamma$, as in Theorem~\ref{thm:pagerank}. 
Hence, given the eigenvector decomposition, each guess $y(\gamma)$ of the 
binary search can be computed by expanding the above series, which requires 
a linear number of inner product computations.
While this may yield a worse running time than Theorem~\ref{thm:comp} in 
the worst case, in the case that the graph is well-approximated by a small 
number $k$ of dominant eigenvectors, then the computation is reduced to 
only $k$ straightforward inner product computations.

\subsection{Proof of Theorem~\ref{thm:pagerank}}
\label{sxn:optimize-proofs-pagerank}

We start with an outline of the proof.
Although the program \textsf{LocalSpectral}$(G,s,\kappa)$  is \emph{not} 
convex, it can be relaxed to the convex semidefinite program 
$\sdp_p(G,s, \kappa)$ of Figure~\ref{fig:sdp}. 
Then, one can observe that strong duality holds for this SDP relaxation. 
Using strong duality and the related complementary slackness conditions, one 
can argue that the primal $\sdp_p(G,s, \kappa)$ has a rank one unique optimal 
solution under the conditions of the theorem.
This implies that the optimal solution of $\sdp_p(G,s, \kappa)$ is the same 
as the optimal solution of \textsf{LocalSpectral}$(G,s, \kappa)$.
Moreover, combining this fact with the complementary slackness condition 
obtained from the dual $\sdp_d(G,s, \kappa)$ of Figure~\ref{fig:sdp}, one 
can derive that the optimal rank one solution is of the form promised by 
Theorem~\ref{thm:pagerank}.

\begin{figure}
\begin{minipage}{0.5\textwidth}
\begin{alignat*}{4}
\quad&  &\text{minimize} \quad &&  L_{G} \circ  X\\
%%% ICML SUB %%%   &  &\text{subject to} \quad &&  L_{K_{n}} \circ  X = 1\ \\
  &  &\text{s.t.} \quad &&  L_{K_{n}} \circ  X = 1\ \\
  &  & &&  ( D_{G} s)( D_{G} s)^T \circ  X \geq \kappa \\
  & & & & X \succeq 0
\end{alignat*} 
\end{minipage}
\begin{minipage}{.25\textwidth}
\begin{alignat*}{4}
\quad&  &\text{maximize} \quad && \alpha + \kappa \beta\\
%%% ICML SUB %%%   &  &\text{subject to} \quad &&  L_{G}  \succeq  \alpha  L_{K_{n}} + \beta ( D {s})( D {s})^T \ \\
  &  &\text{s.t.} \quad &&  L_{G}  \succeq  \alpha  L_{K_{n}} + \beta ( D_{G} {s})( D_{G} {s})^T \ \\
  &  & &&  \beta \geq 0 \\
  & & & & \alpha \in \mathbb{R}
\end{alignat*}
\end{minipage}
\caption{
Left:  Primal SDP relaxation of \textsf{LocalSpectral}$(G,s, \kappa)$:  
$\sdp_{p}(G,s,\kappa)$;  
for this primal, the optimization variable is $X \in \mathbb{R}^{V \times V}$ 
such that $X$ is symmetric and positive semidefinite.
Right: Dual SDP relaxation of \textsf{LocalSpectral}$(G,s, \kappa)$: 
$\sdp_{d}(G,s,\kappa)$;
for this dual, the optimization variables are $\alpha,\beta\in\mathbb{R}$. 
Recall that $L_{K_{n}} \defeq D_{G}-\frac{1}{\vol(G)}D_{G}11^{T}D_{G}.$
}
\label{fig:sdp}
\end{figure}

Before proceeding with the details of the proof, we pause to make several 
points that should help to clarify our approach.
\begin{itemize}
\item
First, since it may seem to some readers to be unnecessarily complex to relax 
\textsf{LocalSpectral} as an SDP, we emphasize that the motivation for 
relaxing it in this way is that we would like to prove 
Theorem~\ref{thm:pagerank}.
To prove this theorem, we must understand the form of the optimal solutions 
to the non-convex program \textsf{LocalSpectral}. 
Thus, in order to overcome the non-convexity, we relax \textsf{LocalSpectral} 
to $\sdp_p(G,s, \kappa)$ (of Figure~\ref{fig:sdp}) by ``lifting'' the 
rank-$1$ condition implicit in \textsf{LocalSpectral}. 
Then, strong duality applies; and it implies a set of sufficient optimality 
conditions. 
By combining these conditions, we will be able to establish that an optimal 
solution $X^{\star}$ to $\sdp_p(G,s, \kappa)$ has rank $1$, \emph{i.e.}, it has the 
form $X^{\star}=x^{\star}x^{\star T}$ for some vector $x^{\star}$; and thus it yields an optimal solution 
to \textsf{LocalSpectral}, \emph{i.e.}, the vector~$x^{\star}$. 
\item
Second, in general, the value of a relaxation like $\sdp_p(G,s, \kappa)$ 
may be strictly less than that of the original program 
(\textsf{LocalSpectral}, in this case). 
Our characterization and proof will imply that the relaxation is tight, 
\emph{i.e.}, that the optimum of $\sdp_p(G,s, \kappa)$ equals that of 
\textsf{LocalSpectral}.
The reason is that one can find a rank-$1$ optimal solution to 
$\sdp_p(G,s, \kappa)$, which then yields an optimal solution of the same 
value for \textsf{LocalSpectral}. 
Note that this also implies that strong duality holds for the non-convex 
\textsf{LocalSpectral}, although this observation is not needed for our 
proof.
%%% WHY DID NV DROP THAT LAST STATEMENT.
\end{itemize}
\noindent
That is, although it may be possible to prove Theorem~\ref{thm:pagerank} in 
some other way that does not involve SDPs, we 
%%%do not know of such a proof.
%%%We 
chose this proof since it is simple and intuitive and correct; and we 
note that Appendix B in the textbook of Boyd and Vandenberghe~\cite{Boyd04} 
proves a similar statement by the same SDP-based approach.

% \begin{figure}
% \begin{minipage}{0.5\textwidth}
% \begin{alignat*}{4}
% \quad&  &\text{minimize} \quad &&  L_{G} \circ  X\\
%   &  &\text{subject to} \quad &&  L_{K_{n}} \circ  X = 1\ \\
%   &  & &&  ( D s)( D s)^T \circ  X \geq \kappa
% \end{alignat*} 
% \end{minipage}
% \begin{minipage}{.25\textwidth}
% \begin{alignat*}{4}
% \quad&  &\text{maximize} \quad && \alpha + \kappa \beta\\
%   &  &\text{subject to} \quad &&  L_{G}  \succeq  \alpha  L_{K_{n}} + \beta ( D {s})( D {s})^T \ \\
%   &  & &&  \beta \geq 0
% \end{alignat*}
% \end{minipage}
% \caption{
% Left:  Primal SDP relaxation of \textsf{LocalSpectral}$(G,s, \kappa)$: 
% $\sdp_{p}(G,s,\kappa)$;  
% for this primal, the optimization variable is $X \in \mathbb{R}^{n \times n}$ 
% such that $X$ is symmetric and positive semidefinite.
% Right: Dual SDP relaxation of \textsf{LocalSpectral}$(G,s, \kappa)$: 
% $\sdp_{d}(G,s,\kappa)$;
% for this dual, the optimization variables are $\alpha,\beta\in\mathbb{R}$.
% XXX.  TO BE REMOVED EVENTUALLY.
% }
% \label{fig:sdp}
% \end{figure} 

%In more detail, the proof of this theorem will 
Returning to the details of the proof, we will 
proceed to prove the theorem by establishing a sequence of claims.
First, consider $\sdp_p(G,s, \kappa)$ and its dual $\sdp_d(G,s, \kappa)$ (as 
shown in Figure~\ref{fig:sdp}). 
The following claim uses the fact that, given $X=xx^T$ for 
$x \in \mathbb{R}^V$, and for any matrix $A \in \mathbb{R}^{V \times V}$, we 
have that $A \circ X = x^T A x$. 
In particular, $L_G \circ X = x^T L_G x$, for any graph $G$, and 
$(x^{T} D_{G} s)^2 = x^T D_{G} s s^T D_{G} x = D_{G}ss^TD_{G} \circ X$.

\begin{claim}
The primal $\sdp_p(G,s, \kappa)$ is a relaxation of the vector program 
\textsf{LocalSpectral}$(G,s, \kappa)$.
\end{claim}
\begin{proof} %%%{[of Claim 1]}
Consider a vector $x$ that is a feasible solution to 
\textsf{LocalSpectral}$(G,s, \kappa)$, and note that $X=xx^T$ is a feasible 
solution to $\sdp_p(G,s, \kappa)$.
\end{proof}

\noindent
Next, we establish the strong duality of $\sdp_p(G,s, \kappa)$.
(Note that the feasibility conditions and
complementary slackness conditions
stated below
may not suffice to 
establish the optimality, in the absence of this claim; 
hence, without this claim, we could not prove the subsequent claims, which 
are needed to prove the theorem.)

\begin{claim}
Strong duality holds between $\sdp_p(G,s, \kappa)$ and $\sdp_d(G,s, \kappa)$.
\end{claim}
\begin{proof} %%%{[of Claim 2]}
Since $\sdp_p(G,s, \kappa)$ is convex, it suffices to verify that Slater's 
constraint qualification condition~\cite{Boyd04} is true for this primal SDP.
Consider $ X= {s}{s}^T$.  
Then, $(D_{G}{s})(D_{G}{s})^T \circ {s}{s}^T = ({s}^T D_{G}{s})^2 = 1>\kappa$.
\end{proof}

\noindent
Next, we use this result to establish the following two claims.
In particular, strong duality allows us to prove the following claim showing 
the KKT-conditions, \emph{i.e.}, the feasibility conditions and 
complementary slackness conditions stated below, suffice to establish
 optimality.

%%% \begin{figure}
%%% \begin{minipage}{0.5\textwidth}
%%% \begin{eqnarray*}
%%%   L_{K_{n}} \circ  X^\star &=& 1 \\%%MWM \label{F1} \\
%%%  ( D {s})( D {s})^T \circ  X^\star &\geq& \kappa  \\%%MWM \label{F2}  \\
%%%   L_{G}- \alpha^\star  L_{K_{n}}  - \beta^\star ( D {s})( D {s})^T &\succeq& 0  \\%%MWM \label{F3} \\
%%%  \beta^\star &\geq& 0  %%MWM \label{F4}  
%%% \end{eqnarray*}
%%% \end{minipage}
%%% \begin{minipage}{0.5\textwidth}
%%% \begin{eqnarray*}
%%%  \alpha^\star(  L_{K_{n}}  \circ  X^\star - 1) &=& 0  \\%%MWM \label{C1} \\
%%%  \beta^\star ( ( D {s})( D {s})^T \circ  X^\star - \kappa) &=& 0  \\%%MWM \label{C2} \\
%%%  X^\star \circ  ( L_{G}- \alpha^\star  L_{K_{n}}  - \beta^\star ( D {s})( D {s})^T ) &=& 0  \\%%MWM \label{C3}  
%%% \end{eqnarray*}
%%% \end{minipage}
%%% \caption{Left: Feasibility conditions. Right: Complementary slackness conditions.}
%%% \label{fig:CS}
%%% \end{figure}
%%% 
%%% \begin{claim}
%%% The feasibility and complementary slackness conditions  for a primal-dual pair 
%%% $ X^\star,\alpha^\star,\beta^\star$ listed in Figure~\ref{fig:CS} are sufficient for them to be an optimal solution.
%%% \end{claim}
%%% 
\begin{claim}   %%%{[same as Claim 3]}
The following feasibility and complementary slackness conditions are 
sufficient for a primal-dual pair 
$ X^\star,\alpha^\star,\beta^\star$ to be an optimal solution.
The feasibility conditions are:
\begin{eqnarray}
  L_{K_{n}} \circ  X^\star &=& 1  \label{F1} \\
 ( D_{G} {s})( D_{G} {s})^T \circ  X^\star &\geq& \kappa   \label{F2}  \\
  L_{G}- \alpha^\star  L_{K_{n}}  - \beta^\star ( D_{G} {s})( D_{G} {s})^T &\succeq& 0  \label{F3} \\
 \beta^\star &\geq& 0  \label{F4}  ,
\end{eqnarray}
and the complementary slackness conditions are:
\begin{eqnarray}
 \alpha^\star(  L_{K_{n}}  \circ  X^\star - 1) &=& 0 \label{C1} \\
 \beta^\star ( ( D_{G} {s})( D_{G} {s})^T \circ  X^\star - \kappa) &=& 0 \label{C2} \\
 X^\star \circ  ( L_{G}- \alpha^\star  L_{K_{n}}  - \beta^\star ( D_{G} {s})( D_{G} {s})^T ) &=& 0 \label{C3}  .
\end{eqnarray}
\end{claim}
\begin{proof}
This follows from the convexity of $\sdp_p(G,s, \kappa)$ and Slater's 
condition~\cite{Boyd04}.
\end{proof}

\begin{claim}
\label{claim:rankone}
These feasibility and complementary slackness conditions, coupled with the 
assumptions of the theorem, imply that  $X^{\star}$ must be 
rank $1$ and $\beta^{\star}>0.$
\end{claim}
\begin{proof}  %%%{[of Claim 4]}
%We start by establishing two simple facts.
%\begin{fact}
%\label{fct:2}
%$\alpha^{\star} \leq  \lambda_{2}(G).$ Moreover if $\lambda_{2}(G)=\alpha^{\star}$ then $ v_{2}^T D s = 0.$
%\end{fact}
%\begin{proof}
Plugging in $v_{2}$ in Equation~\eqref{F3}, we obtain that 
$ v_{2}^{T}L_{G}v_{2} - \alpha^{\star} -\beta^{\star}  (v_{2}^T D_{G} s)^{2} \geq 0.$
But $ v_{2}^{T}L_{G}v_{2}=\lambda_{2}(G)$ and $\beta^{\star} \geq 0.$ Hence, $\lambda_{2}(G) \geq \alpha^{\star}.$ 
Suppose $\alpha^\star = \lambda_2(G).$ As $s^T D_{G} v_2 \neq 0,$ it must be the case that $\beta^\star = 0.$ Hence, by Equation~\eqref{C3}, we must have $X^\star \circ L(G) = \lambda_2(G),$ which implies that $X^\star = v_2v_2^T,$ {\em i.e.},  the optimum for \textsf{LocalSpectral} is the global eigenvector $v_2$. This corresponds to a choice of $\gamma = \lambda_2(G)$ and $c$ tending to infinity.

Otherwise, we may assume that $\alpha^\star< \lambda_2(G).$ Hence, since $G$ is connected and $\alpha^{\star} <\lambda_{2}(G),$ $L_{G}-\alpha^{\star}L_{K_{n}}$ has rank exactly $n-1$ and kernel parallel to the vector $1.$ 
From the complementary slackness condition \eqref{C3} we can deduce that the image of $X^{\star}$ is in the kernel of $ L_{G}- \alpha^\star  L_{K_{n}}  - \beta^\star ( D_{G} {s})( D_{G} {s})^T.$ 
If $\beta^\star > 0,$ we have that $ \beta^\star ( D_{G} {s})( D_{G} {s})^T$ is a rank one matrix and, since $s^T D_{G} 1 = 0,$ it reduces the rank of $L_{G}-\alpha^{\star}L_{K_{n}}$ by one precisely. If $\beta^{\star}=0$ then $X^{\star}$ must be $0$ which is not possible  if $\sdp_{p}(G,s,\kappa)$ is feasible.
Hence, the rank of $ L_{G}- \alpha^\star  L_{K_{n}}  - \beta^\star ( D_{G} {s})( D_{G} {s})^T$ must be exactly $n-2.$ As we may assume that $1$ is in the kernel of $X^\star$, $X^{\star}$ must be of rank one. 
This proves the claim. 
\end{proof}

%\vspace{-2mm}

\noindent
Now we complete the proof of the theorem.
From the claim it follows that, $X^{\star}=x^{\star}x^{\star T}$ where $x^{\star}$ satisfies the equation 
$  (L_{G}- \alpha^\star  L_{K_{n}}  - \beta^\star ( D_{G} {s})( D_{G} {s})^T)x^{\star}=0.$
From the second complementary slackness condition, 
%%% 
%%% in Figure~\ref{fig:CS},  
Equation~\eqref{C2},
%%% 
and the fact that $\beta^{\star}>0,$ we obtain that 
$ (x^{\star})^T D_{G} s = \pm \sqrt{\kappa}.$ 
Thus, 
$x^{\star} =\pm \beta^{\star} \sqrt{\kappa} (L_{G}-\alpha^{\star}L_{K_{n}})^{+}D_{G}s,$ as required.

\section{Application to Partitioning Graphs Locally}
\label{sxn:partition}

In this section, we describe the application of \textsf{LocalSpectral} to 
finding locally-biased partitions in a graph, \emph{i.e.}, to finding 
sparse cuts around an input seed vertex set in the graph.
For simplicity, in this part of the paper, we let the instance graph $G$ be 
unweighted.

\subsection{Background on Global Spectral Algorithms for Partitioning Graphs}
\label{sxn:partition-background}

We start with a brief review of global spectral graph partitioning.
Recall that the basic global graph partitioning problem is: given as input 
a graph $G=(V,E)$, find a set of nodes $S \subseteq V$ to solve 
$$
\phi(G)=\min_{S \subseteq V} \phi(S)  .
$$
Spectral methods approximate the solution to this intractable global problem 
by solving the relaxed problem $\textsf{Spectral}(G)$ presented in 
Figure~\ref{fig:spectral}.
To understand this optimization problem, recall that $x^TL_{G}x$ counts the 
number of edges crossing the cut and that $x^TD_{G}x=1$ encodes a variance 
constraint; thus, the goal of $\textsf{Spectral}(G)$ is to minimize the 
number of edges crossing the cut subject to a given variance. 
Recall that for $T\subseteq V$, we let $1_{T} \in \{0,1\}^{V}$ be a vector 
which is $1$ for vertices in $T$ and $0$ otherwise.
Then for  a cut $(S,\bar{S})$, if we define the vector 
$v_{S} \defeq  \sqrt{\frac{\vol(S) \cdot \vol (\bar{S})}{\vol(G)}} \cdot \left( \frac{1_{S}}{\vol(S)}-\frac{1_{\bar{S}}}{\vol{\bar{S}}}\right)$, 
it can be checked that $v_{S}$ satisfies the constraints of \textsf{Spectral} 
and has objective value $\phi(S)$. 
Thus, $\lambda_{2}(G) \leq \min_{S \subseteq V} \phi(S)= \phi(G)$. 

Hence, $\textsf{Spectral}(G)$ is a relaxation of the minimum conductance 
problem. 
Moreover, this program is a good relaxation in that a good cut can be 
recovered by considering a truncation, \emph{i.e.}, a sweep cut, of the
vector $v_{2}$ that is the optimal solution to $\textsf{Spectral}(G)$.
(That is, \emph{e.g.}, consider each of the $n$ cuts defined by the vector 
$v_{2}$, and return the cut with minimum conductance value.)
This is captured by the following celebrated result often referred to as 
Cheeger's Inequality.
%%%
\begin{theorem}[Cheeger's Inequality]
\label{thm:cheeger1}
For a connected graph $G$, 
$\phi(G) \leq O(\sqrt{ \lambda_{2}(G)})$.
\end{theorem}

\noindent
Although there are many proofs known for this theorem (see, 
\emph{e.g.},~\cite{Chung:1997}), a particularly interesting proof was 
found by Mihail~\cite{Mihail}; this proof involves rounding any \emph{test 
vector} (rather than just the optimal vector), and it achieves the same 
guarantee as Cheeger's Inequality. 
%%%
\begin{theorem}[Sweep Cut Rounding]
\label{thm:cheeger2}
Let $x$ be a vector such that $x^T D_G 1 =0$. Then there is a $t$ for which 
the set of vertices $S := \mathsf{SweepCut}_{t}(x) \defeq \{i: x_{i} \geq t \}$ 
satisfies $ \frac{ x^{T}L_{G}x}{x^{T}D_{G}x} \geq \phi^{2}(S)/8$.
\end{theorem}
%%%
It is the form of Cheeger's Inequality provided by Theorem~\ref{thm:cheeger2}
that we will use below.

\subsection{Locally-Biased Spectral Graph Partitioning}
\label{sxn:partition-local}

Here, we will exploit the analogy between \textsf{Spectral} and 
\textsf{LocalSpectral} by applying the global approach just outlined to the 
following locally-biased graph partitioning problem: 
given as input a graph $G=(V,E)$, an input node $u$, and a positive integer
$k$, find a set of nodes $T \subseteq V$ achieving
$$
\phi(u,k) = \min_{T \subseteq V: u\in T, \vol(T) \le k} \phi(T) .
$$
That is, the problem is to 
find the best conductance set of nodes of volume no greater 
than $k$ that contains the input node $v$. 

As a first step, we show that we can choose the seed set and correlation parameters $s$ and $\kappa$ such 
that $\textsf{LocalSpectral}(G,s,\kappa)$ is a relaxation for this 
locally-biased graph partitioning problem.
\begin{lemma}
\label{lem:relaxation}
For $u \in V$, \textsf{LocalSpectral}$(G,v_{\{u\}},1/k)$ is a relaxation 
of the problem of finding a minimum conductance cut $T$ in $G$ which 
contains the vertex $u$ and is of volume at most~$k$.  
In particular, $\lambda(G,v_{\{u\}},1/k) \leq \phi(u,k)$.
\end{lemma}
\begin{proof}
If we let $x=v_{T}$ in 
\textsf{LocalSpectral}$(G,v_{\{u\}},1/k)$, 
then $v_{T}^{T}L_{G}v_{T}=  \phi(T)$, $v_{T}^{T}D_{G}1=0$, and 
$v_{T}^{T}D_{G}v_{T}=1$.
Moreover, we have that 
$ (v_{T}^{T}D_{G}v_{\{u\}})^{2} = \frac{d_{u}(2m-\vol(T))}{\vol(T) (2m-d_{u})} \geq 1/k$,
which establishes the lemma. 
\end{proof}

\noindent
%%% Just as an optimal vector $x^{\star}$ for \textsf{Spectral} minimizing the 
%%% \emph{mixing} in $G$ (\emph{i.e.}, $x^T L_{G} x$) for a given amount of 
%%% \emph{variance} (\emph{i.e.}, $x^TD_{G}x $), 
%%% with \textsf{LocalSpectral} we want to do the same but subject to the 
%%% additional constraint that the solution vector $x^{\star}$ is well-correlated 
%%% (\emph{i.e.}, $(x^TD_{G}s)^2 \ge \kappa$) with the input seed vector $s$.
%%% %%% DOES THIS BELON SOMEWHERE ELSE BETTER.
%%% 
Next, we can apply Theorem~\ref{thm:cheeger2} to the optimal solution for 
$\textsf{LocalSpectral}(G,v_{\{u\}},1/k)$ and obtain a cut $T$ whose 
conductance is quadratically close to the optimal value 
$\lambda(G,v_{\{u\}},1/k)$. 
By Lemma~\ref{lem:relaxation}, this implies that 
$\phi(T) \leq O(\sqrt{\phi(u,k)})$.
This argument proves the following theorem.
%%%
\begin{theorem}[Finding a Cut]
\label{thm:cut} 
Given an unweighted graph $G=(V,E)$, a vertex $u \in V$ and a positive 
integer $k$, we can find a cut in $G$ of conductance at most 
$O(\sqrt{ \phi(u,k)})$ by computing a sweep cut of the optimal vector for 
$\textsf{LocalSpectral}(G, v_{\{u\}},1/k)$. 
Moreover, this algorithm runs in nearly-linear time in the size of the graph.
\end{theorem}

\noindent
That is, this theorem states that we can perform a sweet cut over the vector 
that is the solution to $\textsf{LocalSpectral}(G,v_{\{u\}},1/k)$ in order 
to obtain a locally-biased partition; and that this partition comes with 
quality-of-approximation guarantees analogous to that provided for the 
global problem $\textsf{Spectral}(G)$ by Cheeger's inequality.

Our final theorem shows that the optimal value of \textsf{LocalSpectral}
also provides a lower bound on the conductance of \emph{other cuts}, as a function 
of how well-correlated they are with the input seed vector.
In particular, when the seed vector corresponds to a cut $U$, this result
allows us to lower bound the conductance of an arbitrary cut $T$, in terms
of the correlation between $U$ and $T$.
The proof of this theorem 
% (which we omit due to lack of space) 
also uses in an essential manner the
duality properties that were used in the 
proof of Theorem~\ref{thm:pagerank}.

\begin{theorem}[Cut Improvement]
\label{thm:improve}
Let $G$ be a  graph and $s \in \mathbb{R}^{n}$ be such that
$s^{T}D_{G} 1=0,$ where $D_{G}$ is the degree matrix of $G.$
In addition, let $\kappa \geq 0$ be a correlation parameter.
Then, for all sets $T \subseteq V$ such that
$\kappa' \defeq  (s^{T}D_{G}v_{T})^{2}$, we have that
%%%
%%% $\phi(T) \geq \lambda(G,s,\kappa)$ if $\kappa \leq \kappa'$ and
%%% $\phi(T) \geq \nfrac{\kappa'}{\kappa} \cdot \lambda(G,s,\kappa)$ if
%%% $\kappa' \leq \kappa.$
%%%
\[
\phi(T) \geq \left\{ \begin{array}{ll}
                       \lambda(G,s,\kappa)
                       & \mbox{if $\kappa \leq \kappa'$} \\
                       \nfrac{\kappa'}{\kappa} \cdot \lambda(G,s,\kappa)
                       & \mbox{if $\kappa' \leq \kappa$.}
                    \end{array}
            \right.
\]
% In particular, if $s=s_{U}$ for some $U \subseteq V,$ then note that
% $\kappa'= K(U,T).$
\end{theorem}
\begin{proof}
It follows from Theorem \ref{thm:pagerank} that $\lambda(G,s,\kappa)$ is the 
same as the optimal value of $\sdp_{p}(G,s,\kappa)$ which, by strong 
duality, is the same as the optimal value of $\sdp_{d}(G,s,\kappa)$. 
Let $\alpha^{\star},\beta^{\star}$ be the optimal dual values to 
$\sdp_{d}(G, s,\kappa).$ 
Then, from the dual feasibility constraint 
$ L_{G}- \alpha^\star  L_{K_{n}}  - \beta^\star ( D_G {s})( D_G {s})^T \succeq 0 ,$
it follows that 
$$ s_{T}^{T}L_{G}s_{T} - \alpha^{\star}s_{T}^{T}L_{K_{n}}s_{T}-\beta^{\star} (s^T D_G s_{T})^{2} \geq 0.$$
Notice that since $ s_{T}^T D_G 1 =0$, it follows that 
$ s_{T}^{T}L_{K_{n}}s_{T}=s_{T}^{T}D_Gs_{T}=1$.
Further, since $s_{T}^{T}L_{G}s_{T}=\phi(T),$ we obtain, if $\kappa \leq \kappa',$ that  
$$
\phi(T) \geq \alpha^{\star} + \beta^{\star} ( s^T D_G s_{T} )^{2} \geq  \alpha^{\star} + \beta^{\star}\kappa = \lambda (G,s,\kappa).
$$
If on the other hand, $\kappa' \leq \kappa,$ then
$$\phi(T) \geq \alpha^{\star} + \beta^{\star} (s^T D_G s_{T}) ^{2} \geq \alpha^{\star} + \beta^{\star}\kappa \geq   \nfrac{ \kappa'}{\kappa} \cdot (\alpha^{\star} + \beta^{\star}\kappa) =\nfrac{ \kappa'}{\kappa} \cdot  \lambda (G,s,\kappa).
$$
% Finally, observe that if $s=s_{U}$ for some $U \subseteq V,$ then $ s_{U}^T D_G s_{T} )^{2} =K(U,T).$ 
Note that strong duality was used here. 
\end{proof}

\noindent
Thus, although the relaxation guarantees of Lemma~\ref{lem:relaxation} only 
hold when the seed set is a single vertex, we can use 
Theorem~\ref{thm:improve} to consider the following problem: given a graph 
$G$ and a cut $(T,\bar{T})$ in the graph, find a cut of minimum conductance 
in $G$ which is well-correlated with $T$ or certify that there is none. 
Although one can imagine many applications of this primitive, the main 
application that motivated this work was to explore  clusters nearby or 
around a given \emph{seed set} of nodes in data graphs.  
This will be illustrated in our empirical evaluation 
in~Section~\ref{sxn:empirical}.

\subsection{Our Geometric Notion of Correlation Between Cuts}
\label{sxn:partition-geometric}

Here we pause to make explicit the geometric notion of correlation between 
cuts (or partitions, or sets of nodes) that is used by \textsf{LocalSpectral}, 
and that has already been used in various guises in previous sections.
Given a cut $(T,\bar{T})$ in a graph $G=(V,E)$, a natural vector in 
$\mathbb{R}^{V}$ to associate with it is its characteristic vector, in which
case the correlation between a cut $(T,\bar{T})$ and another cut 
$(U,\bar{U})$ can be captured by the inner product of the characteristic 
vectors of the two cuts. 
A somewhat more refined vector to associate with a cut is the vector 
obtained after removing from the characteristic vector its projection along 
the all-ones vector.  
In that case, again, a notion of correlation is related to the inner product 
of two such vectors for two cuts. 
More precisely, given a set of nodes $T \subseteq V$, or equivalently a cut 
$(T,\bar{T})$, one can define the unit vector $s_{T}$~as
\[ 
s_{T}(i) = \left\{ \begin{array}{ll}
                      \sqrt{\nfrac{\vol(T)\vol(\bar{T})}{2m}} \cdot \nfrac{1}{\vol(T)}  & \mbox{if $i \in T $} \\
                      - \sqrt{\nfrac{\vol(T)\vol(\bar{T})}{2m}} \cdot \nfrac{1}{\vol(\bar{
T})}    & \mbox{if $i \in \bar{T}$}   .
                   \end{array}
           \right. \]
That is, 
$
s_T \defeq  \sqrt{\frac{\vol(T)\vol(\bar{T})}{2m}} \; \left(\frac{1_T}{\vol(T)} - \frac{1_{\bar{T}}}{\vol(\bar{T})}\right)   ,
$
which is exactly the vector defined in Section~\ref{sxn:partition-background}.
It is easy to check that this is well defined: one can replace $s_{T}$ by 
$s_{\bar{T}}$ and the correlation remains the same with any other set. 
Moreover, several observations are immediate.
First, defined this way, it immediately follows that 
$ s_T^T D_G 1 = 0$ and that $ s_T^T D_G s_T = 1$. 
Thus, $s_T \in \mathcal{S}_{D}$ for $T \subseteq V$, where we denote by 
$\mathcal{S}_{D}$ the set of vectors $\{x \in \mathbb{R}^V: x^T D_G 1 = 0\}$;
and $s_T$ can be seen as an appropriately normalized version of the vector 
consisting of the uniform distribution over $T$ minus the uniform 
distribution over $\bar{T}$.%
\footnote{Notice also that $s_T = - s_{\bar{T}}$.  Thus, since 
we only consider quadratic functions of $s_T,$ we 
can consider both $s_T$ and $s_{\bar{T}}$ to be representative vectors for 
the cut $(T, \bar{T}).$ }
Second, one can introduce the following measure of correlation between two 
sets of nodes, or equivalently between two cuts, say a cut $(T, \bar{T})$ 
and a cut $(U, \bar{U})$:
$$
K(T,U) \defeq ( s_T D_G s_U )^2 .
$$
The proofs of the following simple facts regarding $K(T,U)$ are omitted:
$K(T,U) \in [0,1]$;
$K(T,U) = 1$ if and only if $T=U$ or $\bar{T}=U$;
$K(T,U) = K(\bar{T}, U)$; and
$K(T,U) = K(T, \bar{U})$.
Third, although we have described this notion of geometric correlation in 
terms of vectors of the form $s_T \in \mathcal{S}_{D}$ that represent 
partitions $(T,\bar{T})$, this correlation is clearly well-defined for other 
vectors $s \in \mathcal{S}_{D}$ for which there is not such a simple 
interpretation in terms of cuts.
Indeed, in Section~\ref{sxn:optimize} we considered the case that $s$ was 
an arbitrary vector in $\mathcal{S}_{D}$, while in the first part of 
Section~\ref{sxn:partition-local} we considered the case that $s$ was the 
seed set of a single node.
In our empirical evaluation in Section~\ref{sxn:empirical}, we will 
consider both of these cases as well as the case that $s$ encodes the 
correlation with cuts consisting of multiple nodes.

\section{Empirical Evaluation}
\label{sxn:empirical} 

In this section, we provide an empirical evaluation of \textsf{LocalSpectral} 
by illustrating its use at finding and evaluating locally-biased 
low-conductance cuts, \emph{i.e.}, sparse cuts or good clusters, around an input seed set of nodes in 
a data graph.
We start with a brief discussion of a very recent and pictorially-compelling 
application of our method to a computer vision problem; and then we discuss 
in detail how our method can be applied to identify clusters and communities 
in a more heterogeneous and more difficult-to-visualize social network 
application.

\subsection{Semi-Supervised Image Segmentation}

Subsequent to the initial dissemination of the technical report version of 
this paper, Maji, Vishnoi, and Malik~\cite{MVM11} applied our methodology to 
the problem of finding locally-biased cuts in a computer vision application.
Recall that image segmentation is the problem of partitioning a digital image 
into segments corresponding to significant objects and areas in the image. 
A standard approach consists in converting the image data into a similarity 
graph over the the pixels and applying a graph partitioning algorithm to 
identify relevant segments. 
In particular, spectral methods have been popular in this area since the 
work of Shi and Malik~\cite{ShiMalik00_NCut}, which used the second 
eigenvector of the graph to approximate the so-called normalized cut 
(which, recall, is an objective measure for image segmentation that is 
practically equivalent to conductance). 
However, a difficulty in applying the normalized cut method is that in many 
cases global eigenvectors may fail to capture important local segments of 
the image.
The reason for this is that they aggressively optimize a global objective 
function and thus they tend to combine multiple segments together; 
this is illustrated pictorially in the first row of Figure~\ref{fig:imseg}.

This difficulty can be overcome in a semi-supervised scenario by using 
our \textsf{LocalSpectral} method. 
Specifically, one often has a small number of ``ground truth'' labels that 
correspond to known segments, and one is interested in extracting and 
refining the segments in which those labels reside.
In this case, if one considers an input seed corresponding to a small number 
of pixels within a target object, then \textsf{LocalSpectral} can recover 
the corresponding segment with high precision.
This is illustrated in the second row of Figure~\ref{fig:imseg}. 
This computer vision application of our methodology was motivated by a 
preliminary version of this paper, and it was described in detail and evaluated against competing 
algorithms by Maji, Vishnoi, and Malik~\cite{MVM11}. 
In particular, they show that \textsf{LocalSpectral} achieves a performance superior to 
that of other semi-supervised segmentation algorithms~\cite{YS01,EOK07}; and 
they also show how \textsf{LocalSpectral} can be incorporated in an 
unsupervised segmentation pipeline by using as input seed distributions 
obtained by an object-detector algorithm~\cite{bmbm10}.

\begin{figure}[h]
%\begin{figure}[t]
   \begin{center}
   \includegraphics[ scale=1.20]{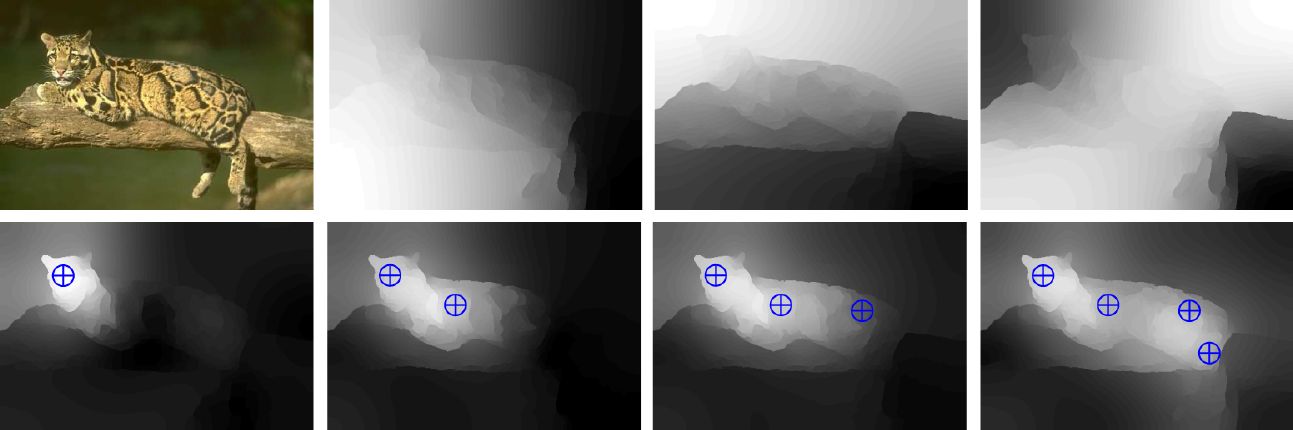}
   \end{center}
\caption{The first row shows the input image and the three smallest eigenvectors of the Laplacian of the corresponding similarity graph computed using the intervening contour cue~\cite{MAFM08}.  Note that no sweep cut of these eigenvectors reveals the leopard.  The second row shows the results of \textsf{LocalSpectral} with a setting of $\gamma = -10 \lambda_2(G)$ with the seed pixels highlighted by crosshairs.  Note how one can to recover the leopard by using a seed vector representing a set of only 4 pixels.  In addition, note how the first seed pixel allows us to capture the head of the animal, while the other seeds help reveal other parts of its body. }
\label{fig:imseg}	
\end{figure}

\subsection{Detecting Communities in Social Networks}

Finding local clusters and meaningful locally-biased communities is also of 
interest in the analysis of large social and information networks.  
A standard approach to finding clusters and communities in many network 
analysis applications is to formalize the idea of a good community with an 
``edge counting'' metric such as conductance or modularity and then to use a 
spectral relaxation to optimize it
approximately~\cite{newman2006finding,newman2006_ModularityPNAS}.
For many very large social and information networks, however, there simply do
not exist good large global clusters, but there do exist small
meaningful local clusters that may be thought of as being nearby 
prespecified seed sets of
nodes~\cite{LLDM08_communities_CONF,LLDM09_communities_IM,LLM10_communities_CONF}.
In these cases, a local version of the global spectral partitioning
problem is of interest, as was shown by Leskovec, Lang, and
Mahoney~\cite{LLM10_communities_CONF}.
Typical networks are very large and, due to their expander-like properties, 
are not 
easily-visualizable~\cite{LLDM08_communities_CONF,LLDM09_communities_IM}.
Thus, in order to illustrate the empirical behavior of our
\textsf{LocalSpectral} methodology in a ``real'' network application related 
to the one that motivated this 
work~\cite{LLDM08_communities_CONF,LLDM09_communities_IM,LLM10_communities_CONF}, 
we examined a small ``coauthorship network'' of scientists.
This network was previously used by Newman~\cite{newman2006finding} to study 
community structure in small social and information networks.

\begin{figure}[h] 
%\begin{figure}[t] 
\begin{center}
\includegraphics[viewport= 0 0 1450 818, clip=true, scale=.27]{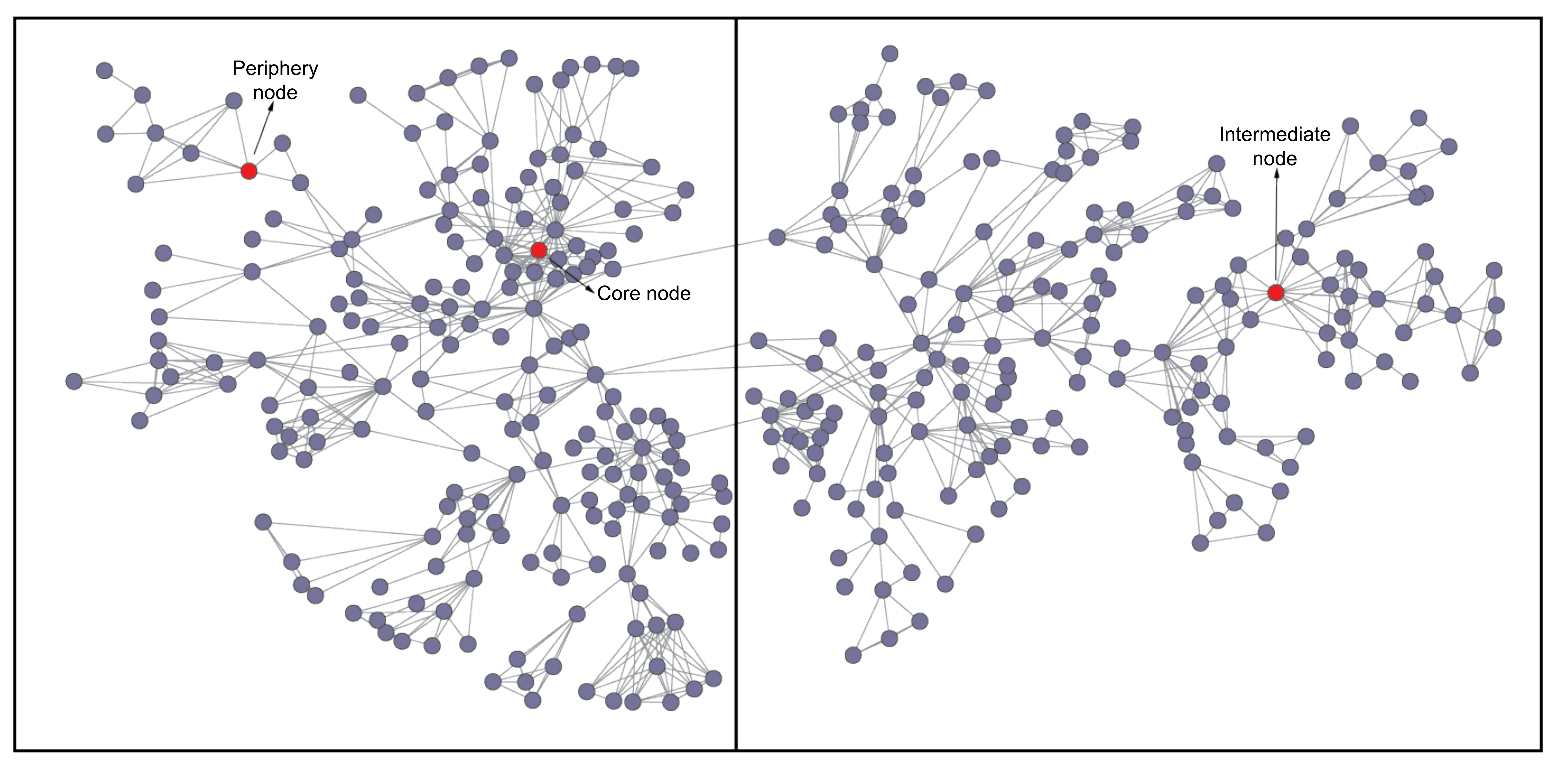} 
\end{center}
\caption{[Best viewed in color.]  The coauthorship network of Newman~\cite{newman2006finding}. This layout was obtained in the Pajek~\cite{Pajek03} visualization software, using the Kamada-Kawai method~\cite{Kamada89} on each component of a partition provided by \textsf{LocalCut} and tiling the layouts at the end.  Boxes show the two main global components of the network, which are displayed separately in subsequent~figures.}
\label{fig:network}
\end{figure}

The corresponding graph $G$ is illustrated in Figure~\ref{fig:network} and
consists of $379$ nodes and $914$ edges, where each node represents an 
author and each unweighted edge represents a coauthorship relationship.
% This graph a standard benchmark in the community detection literature; and
% it combines certain structural heterogeneities that are typical of larger 
% social graphs with a size that permits a more detailed and comprehensive 
% analysis of our method.
% 
The spectral gap $\lambda_2(G) = 0.0029$; and a sweep cut of the eigenvector 
corresponding to this second eigenvalue yields the globally-optimal spectral 
cut separating the graph into two well-balanced partitions, corresponding to 
the left half and the right half of the network, as shown in 
Figure~\ref{fig:network}.
Our main empirical observations, described in detail in the remainder of 
this section, are the following.
\begin{itemize}
\item
First, we show how varying the teleportation parameter allows us to detect 
low-conductance cuts of different volumes that are locally-biased around a 
prespecified seed vertex; and how this information, aggregated over 
multiple choices of teleportation, can improve our understanding of the 
network structure in the neighborhood of the seed.
\item
Second, we demonstrate the more general usefulness of our definition of a 
\emph{generalized} Personalized PageRank vector (where the $\gamma$ 
parameter in Eqn.~(\ref{eqn:xstar}) can be $\gamma\in(-\infty,\lambda_2(G)$) 
by displaying specific instances in which that vector is more effective than 
the usual Personalized PageRank (where only positive teleportation 
probabilities are allowed and thus where $\gamma$ must be negative).
We do this by detecting a wider range of low-conductance cuts at a given 
volume and by interpolating smoothly between very locally-biased solutions 
to \textsf{LocalSpectral} and the global solution provided by the 
\textsf{Spectral} program.
\item
Third, we demonstrate how our method can find low-conductance cuts that are 
well-correlated to more general input seed vectors by demonstrating an 
application to the detection of sparse peripheral regions, \emph{e.g.}, 
regions of the network that are well-correlated with low-degree nodes.
This suggests that our method may find applications in leveraging 
feature data, which are often associated with the vertices of a 
data graph, to find interesting and meaningful cuts.
\end{itemize}
We emphasize that the goal of this empirical evaluation is to illustrate 
how our proposed methodology can be applied in real applications; and thus 
we work with a relatively easy-to-visualize example of a small social graph.
This will allow us to illustrate how the ``knobs'' of our proposed method 
can be used in practice.
In particular, the goal is not to illustrate that our method or heuristic 
variants of it or other spectral-based methods scale to much larger 
graphs---this latter fact is by now 
well-established~\cite{andersen06seed,LLDM08_communities_CONF,LLDM09_communities_IM,LLM10_communities_CONF}.

\subsubsection{Algorithm Description and Implementation}

We refer to our cut-finding algorithm, which will be used to guide our 
empirical study of finding and evaluating cuts around an input seed set of 
nodes and which is a straightforward extension of the algorithm referred 
to in  Theorem~\ref{thm:cut}, as \textsf{LocalCut}.
In addition to the graph, the input parameters for \textsf{LocalCut} are a 
seed vector $s$ (\emph{e.g.}, corresponding to a single vertex $v$), a 
teleportation parameter $\gamma$, and (optionally) a size factor $c$.
Then, \textsf{LocalCut} performs the following steps.
\begin{itemize}
\item
First, compute the vector $x^{\star}$ of Eqn.~(\ref{eqn:xstar}) with seed $s$ 
and teleportation $\gamma$. 
\item
Second, either
perform a sweep of the vector $x^{\star}$, \emph{e.g.}, consider each of the
$n$ cuts defined by the vector and return the the minimum conductance cut 
found along the sweep;
or 
%Alternatively, 
consider only sweep cuts along the vector $x^{\star}$ of volume at most 
$c \cdot k_{\gamma}$, 
where $k_{\gamma}=1/\kappa_{\gamma}$, 
that contain the input vertex $v$, and return the 
minimum conductance cut among such cuts.
\end{itemize}

\noindent
By Theorem~\ref{thm:pagerank}, the vector computed in the first step of 
\textsf{LocalCut}, $x^\star$, is an optimal solution to 
\textsf{LocalSpectral}$(G,s,\kappa_{\gamma})$ for some choice of 
$\kappa_{\gamma}$.
(Indeed, by fixing the above parameters, the $\kappa$ parameter is fixed 
implicitly.)
Then, by Theorem~\ref{thm:cut}, when the vector $x^\star$ is rounded 
(to, \emph{e.g.}, $\{-1,+1\}$) by performing the sweep cut, provably-good 
approximations are guaranteed. 
In addition, when the seed vector corresponds to a single vertex $v$, it
follows from Lemma~\ref{lem:relaxation} that $x^\star$ yields a lower 
bound to the conductance of cuts that contain $v$ and have less than a 
certain volume $k_\gamma$.

Although the full sweep-cut rounding does not give a specific guarantee on 
the volume of the output cut, empirically we have found that it is often 
possible to find small low-conductance cuts in the range dictated by $k_\gamma$.
Thus, in our empirical evaluation, we also consider volume-constrained sweep 
cuts (which departs slightly from the theory but can be useful in practice).
%%%%%%%%%%%%%%%%
That is, we also introduce a new input 
parameter, a \textit{size factor} $c > 0$, that regulates the maximum volume 
of the sweep cuts considered when $s$ represents a single vertex.
In this case, \textsf{LocalCut} does not consider all $n$ cuts defined by 
the vector $x^{\star}$, but instead it considers only sweep cuts of volume 
at most $c\cdot k_\gamma$ that contain the vertex $v$.
(Note that it is a simple consequence of our optimization characterization 
that the optimal vector has sweep cuts of volume at most $k_\gamma$ 
containing $v$.)
This new input parameter turns out to be extremely useful in exploring cuts 
at different sizes, as it neglects sweep cuts of low conductance at large 
volume and allows us to pick out more local cuts around the seed vertex.

In our first two sets of experiments, summarized in 
Sections~\ref{sec:teleport} and~\ref{sec:size}, we used single-vertex seed 
vectors, and we analyzed the effects of varying the parameters $\gamma$ and 
$c$, as a function of the location of the seed vertex in the input graph.
In the last set of experiments, presented in Section~\ref{sec:multi}, we 
considered more general seed vectors, including both seed vectors that 
correspond to multiple nodes, \emph{i.e.}, to cuts or partitions in the 
graph, as well as seed vectors that do not have an obvious interpretation in 
terms of input cuts.
We implemented our code in a combination of MATLAB and C++, solving linear 
systems using the Stabilized Biconjugate Gradient Method~\cite{bicg92} 
provided in MATLAB 2006b. 
On this particular coauthorship network, and on a Dell PowerEdge 1950 
machine with 2.33 GHz and 16GB of RAM, the algorithm ran in less than a few 
seconds.

\subsubsection{Varying the Teleportation Parameter}
\label{sec:teleport}

Here, we evaluate the effect of varying the teleportation parameter 
$\gamma \in (-\infty,\lambda_2(G))$, where recall $\lambda_2(G) = 0.0029$.
Since it is known that large social and information networks are quite 
heterogeneous and exhibit a very strong ``nested core-periphery'' 
structure~\cite{LLDM08_communities_CONF,LLDM09_communities_IM,LLM10_communities_CONF}, 
we perform this evaluation by considering the behavior of \textsf{LocalCut} 
when applied to three types of seed nodes, examples of which are the 
highlighted vertices in Figure~\ref{fig:network}.
These three nodes were chosen to represent three different types of nodes
seen in larger networks: 
a \textit{periphery-like node}, which belongs to a lower-degree and less 
expander-like part of the graph, and which tends to be surrounded by 
lower-conductance cuts of small volume; 
a \textit{core-like node}, which belongs to a denser and higher-conductance 
or more expander-like part of the graph; and 
an \textit{intermediate node}, which belongs to a regime between the 
core-like and the periphery-like regions.

For each of the three representative seed nodes, we executed $1000$ runs of 
\textsf{LocalCut} with $c = 2$ and $\gamma$ varying by $0.001$ increments. 
Figure~\ref{fig:teleportation} displays, for each of these three seeds, a 
plot of the conductance as a function of volume of the cuts found by each 
run of \textsf{LocalCut}. 
We refer to this type of plot as a \textit{local profile plot} since it is 
a specialization of the \textit{network community profile plot}~\cite{LLDM08_communities_CONF,LLDM09_communities_IM,LLM10_communities_CONF} 
to cuts around the specified seed vertex.
In addition, Figure~\ref{fig:teleportation} also plots several other 
quantities of interest:
first, the volume and conductance of the theoretical lower bound yielded by 
each run; 
second, the volume and conductance of the cuts defined by the shortest-path 
balls (in squares and numbered according to the length of the path) around 
each seed (which should and do provide a sanity-check upper bound); 
third, next to each of the plots, we present a color-coded image of 
representative cuts detected by \textsf{LocalCut}; and 
fourth, for each of the cuts illustrated on the left, a color-coded triangle 
and the numerical value of $-\gamma$ is shown on the~right.

\newcommand{\imscale}{.13}
\newcommand{\plscale}{.19}
\renewcommand{\imscale}{.2}
\renewcommand{\plscale}{.27}

\begin{figure}[] 
%\begin{tabular}{cc} 
\subfigure[Selected cuts and profile plot for the \emph{core-like node}.]{
\includegraphics[viewport= 50 -90 818 818, clip=true, scale=\imscale]{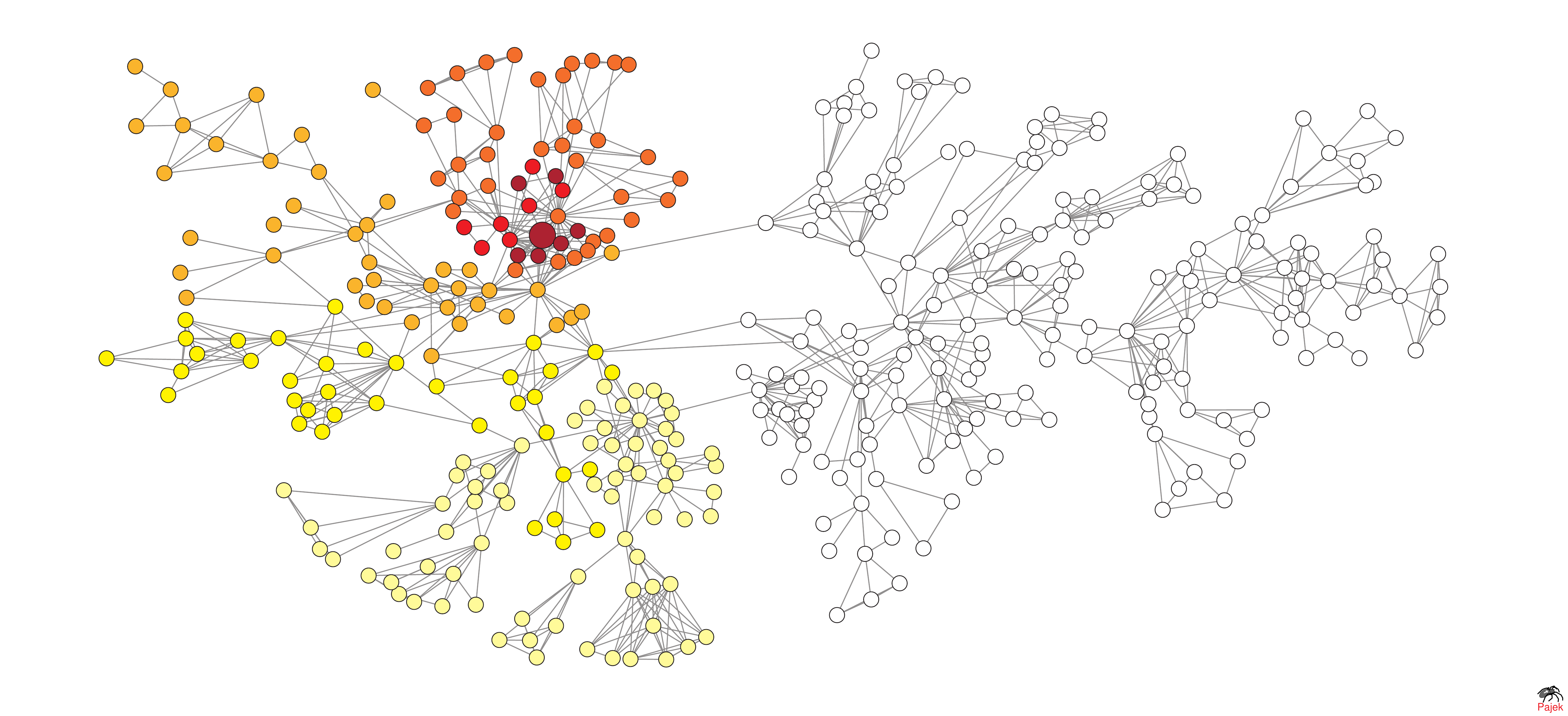} 
\label{subfig:tc}
\hspace{.7in}
%&
\includegraphics[ scale=\plscale]{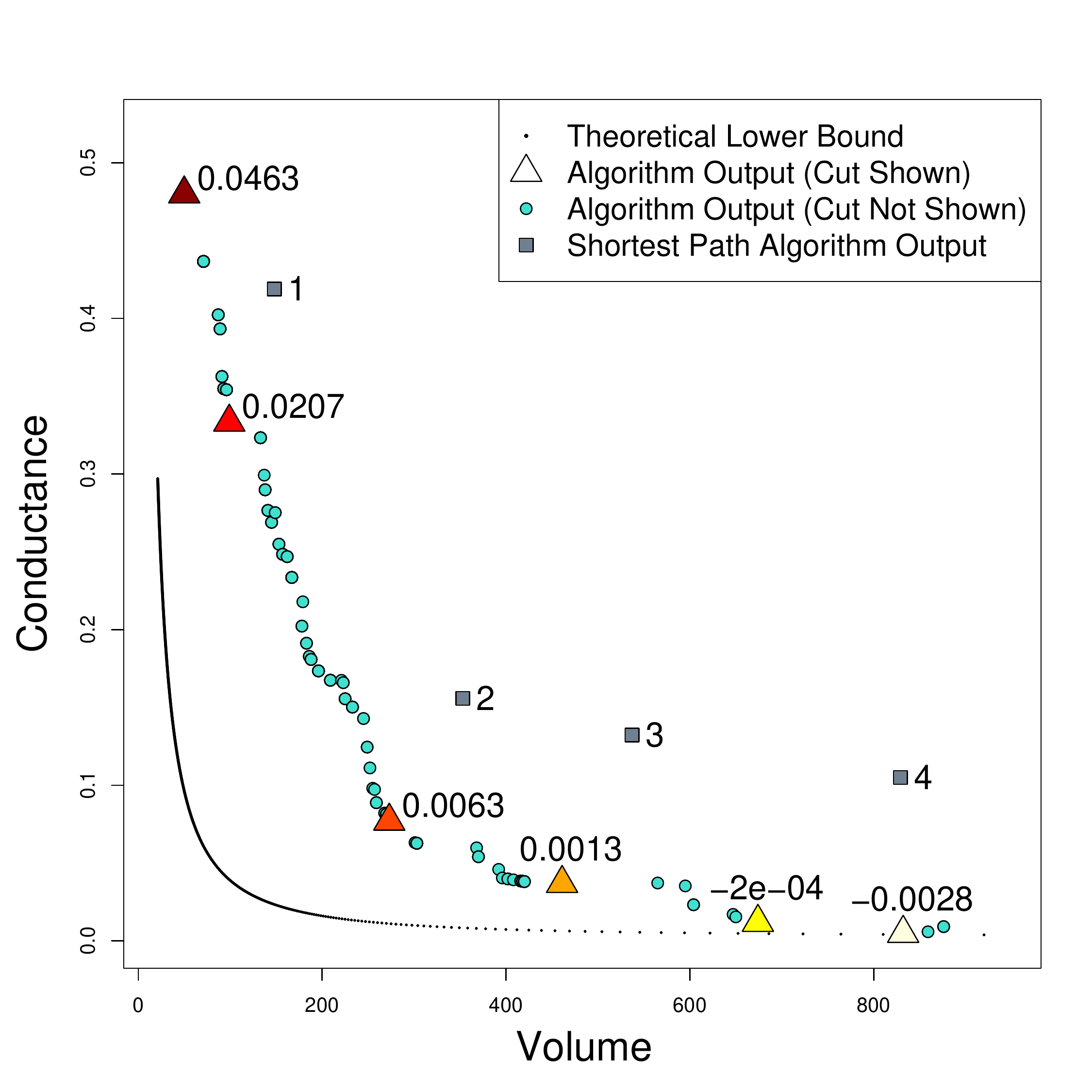}
}
%\\
\subfigure[Selected cuts and profiles plot for the \emph{intermediate node}.]
{
\includegraphics[viewport= 650 -90 1420 818, clip=true, scale=\imscale]{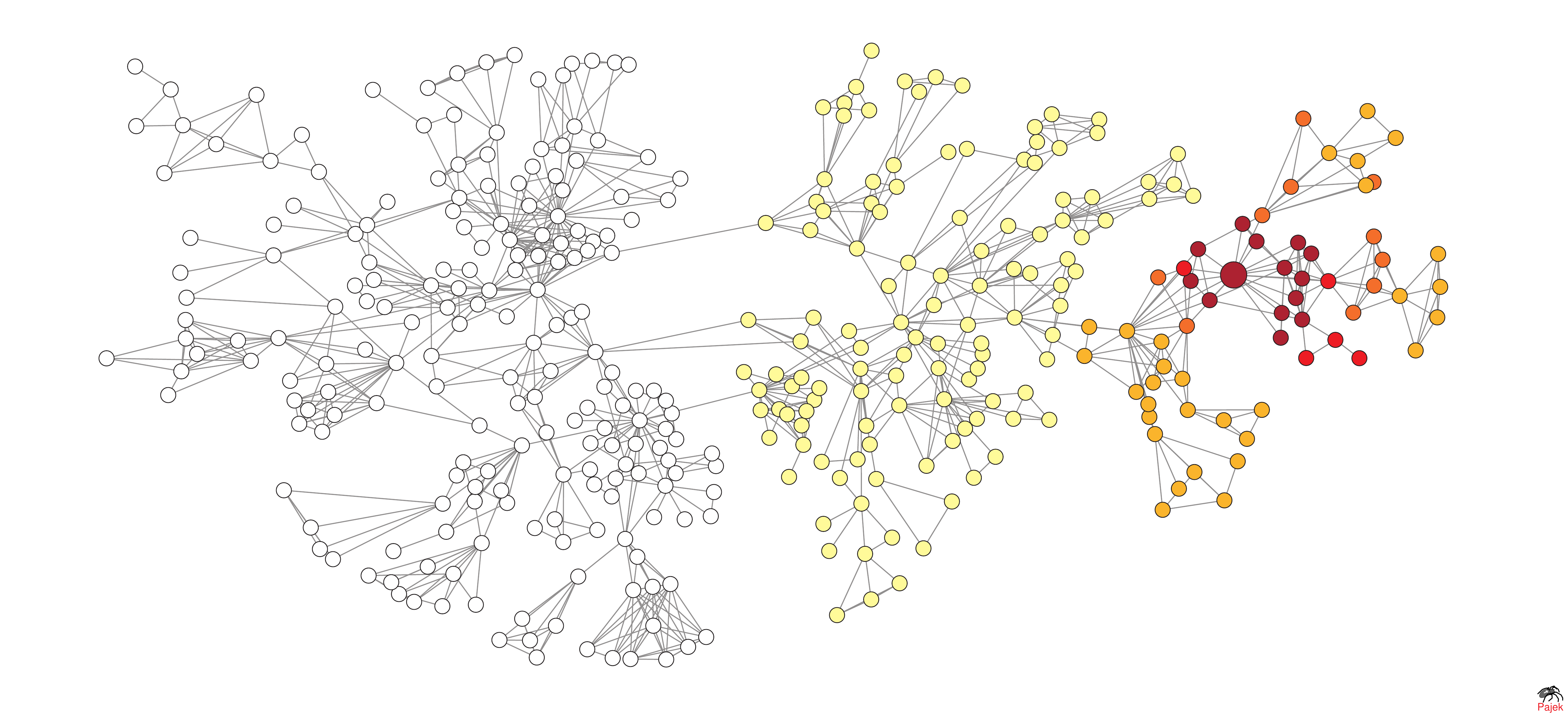} 
\label{subfig:tb}
\hspace{.7in}
%&
\includegraphics[ scale=\plscale]{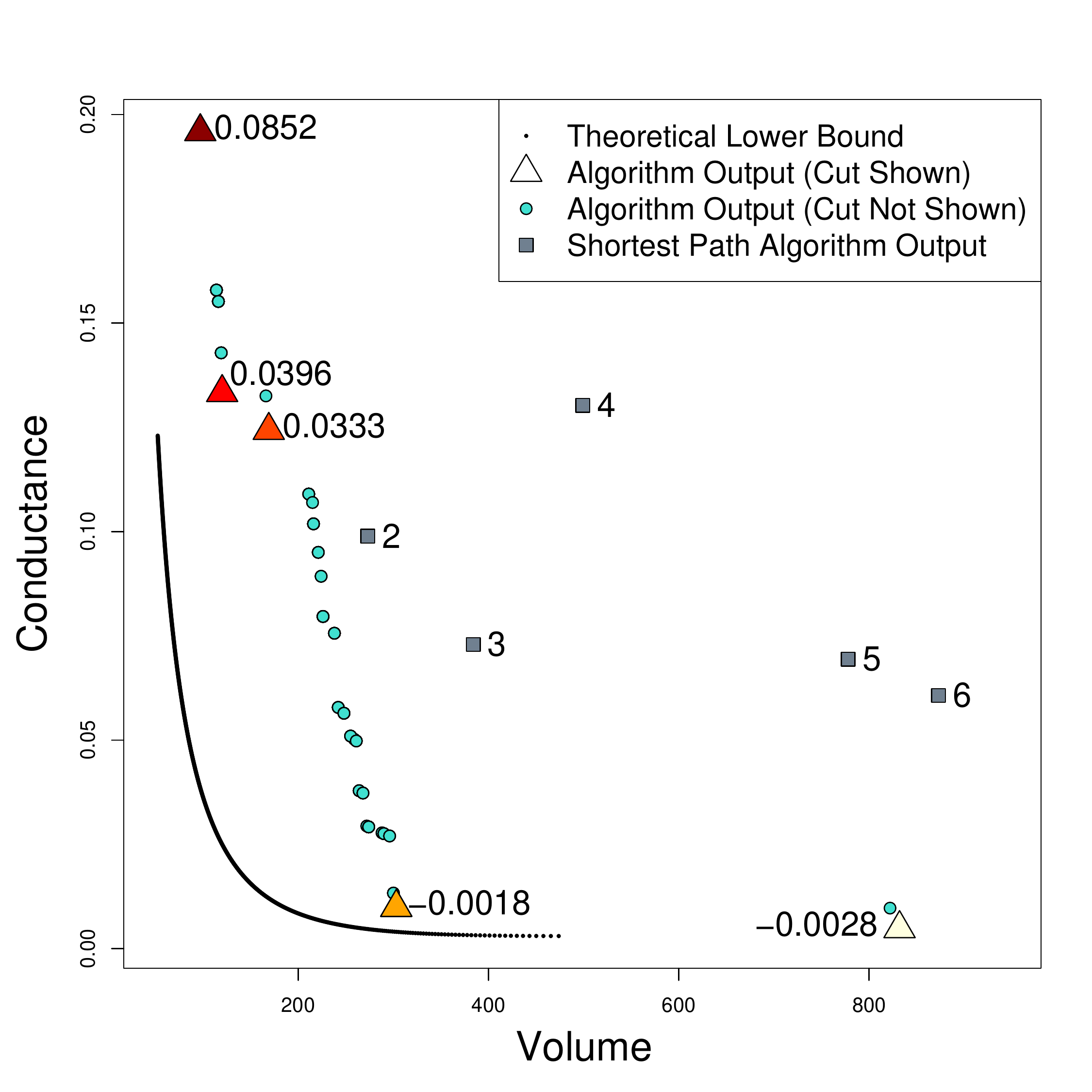}
}
%\\
\subfigure[Selected cuts and profile plot for the \emph{periphery-like node}.]{
\includegraphics[viewport= 50 -90 818 818, clip=true, scale=\imscale]{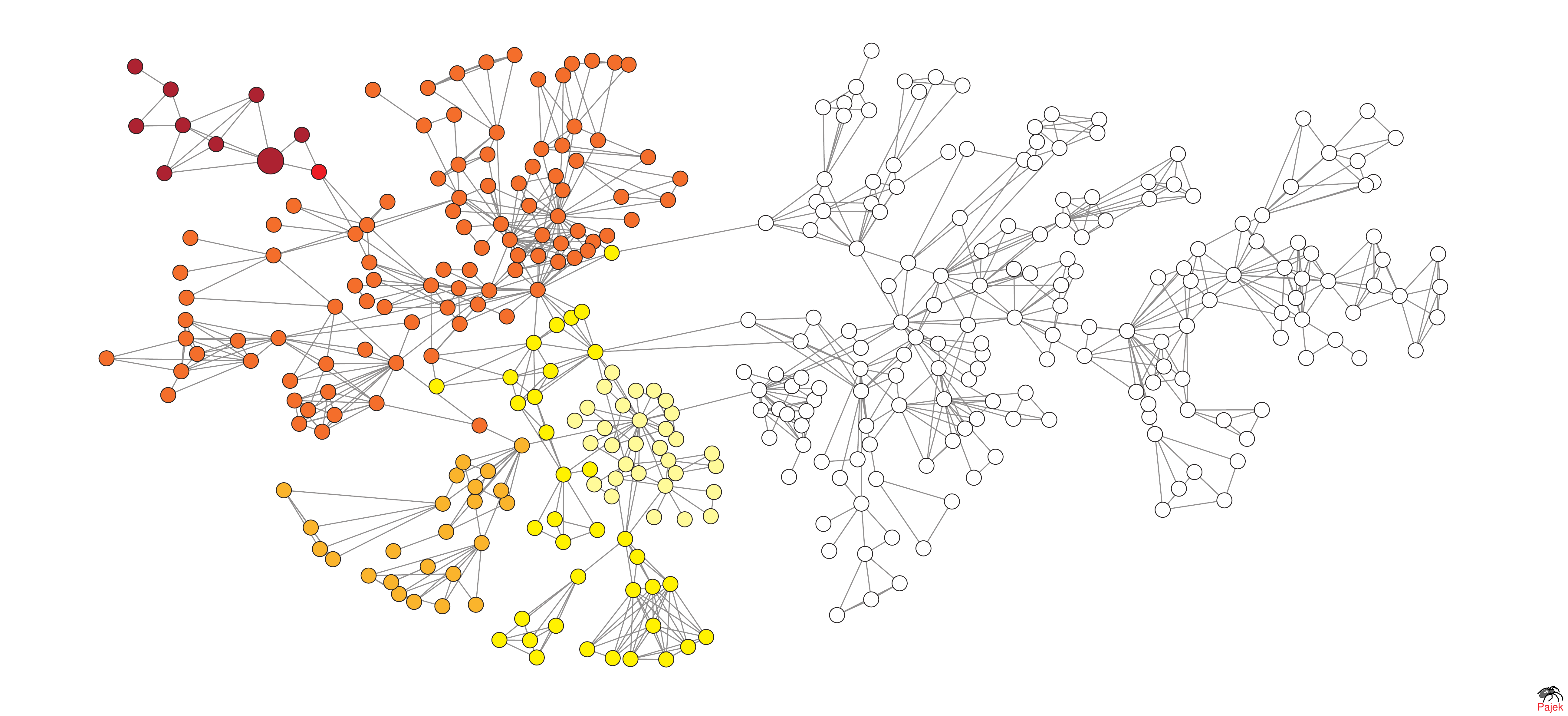}
\label{subfig:ta}
\hspace{.7in}
%&
\includegraphics[ scale=\plscale]{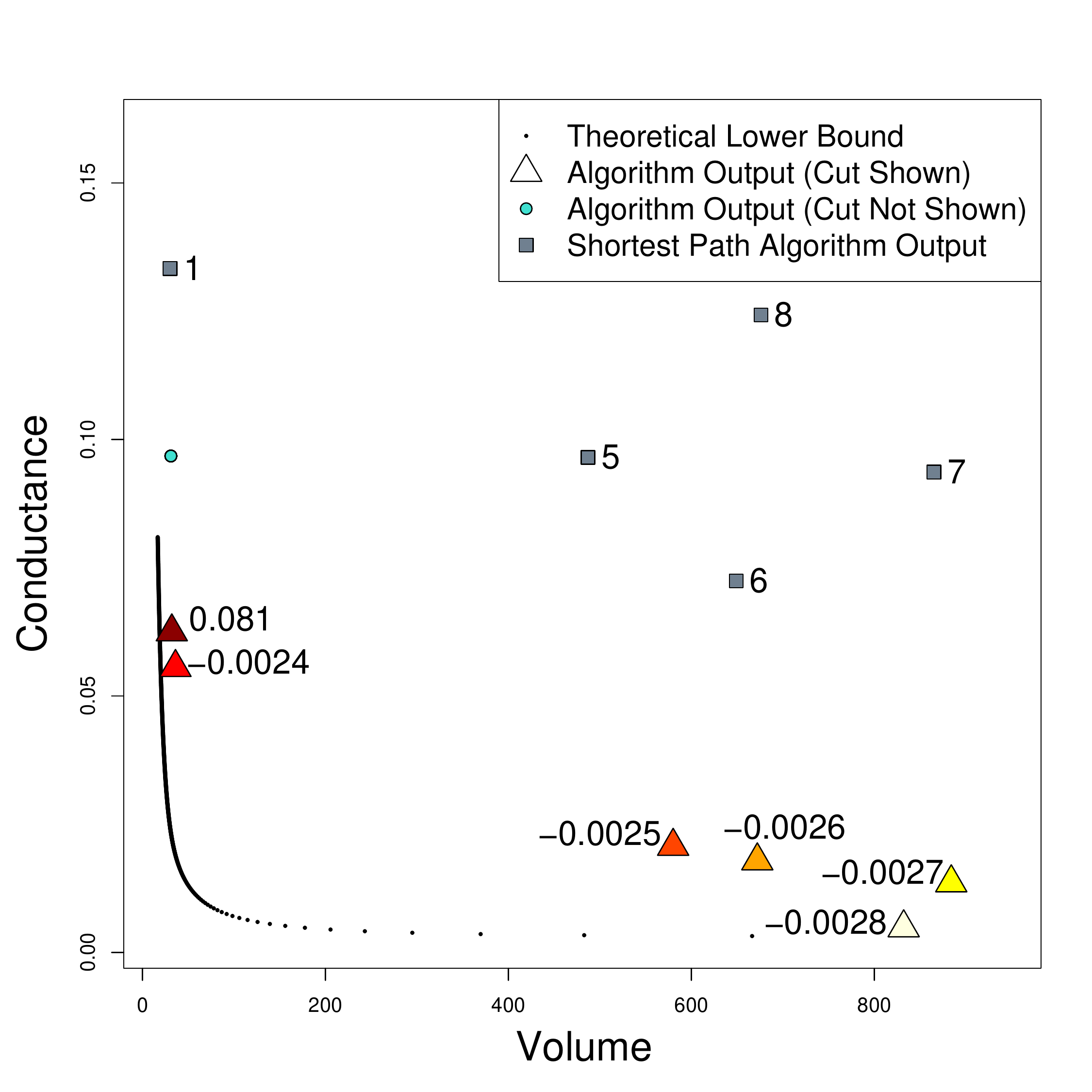}
}
%\\
%\end{tabular}
\caption{[Best viewed in color.]  
Selected cuts and local profile plots for varying $\gamma$. 
The cuts on the left are displayed by assigning to each vertex a color 
corresponding to the smallest selected cut in which the vertex was included. 
Smaller cuts are darker, larger cuts are lighter; and the seed vertex is 
shown slightly larger. 
Each profile plot on the right shows results from $1000$ runs of 
\textsf{LocalCut}, with $c=2$ and $\gamma$ decreasing in $0.001$ increments 
starting at $0.0028$.
For each color-coded triangle, corresponding to a cut on the left, $-\gamma$ 
is also listed.
}
\label{fig:teleportation}
\end{figure}

Several points about the behavior of the \textsf{LocalCut} algorithm as a 
function of the location of the input seed node and that are illustrated in 
Figure~\ref{fig:teleportation} are worth emphasizing.
\begin{itemize}
\item
% CORE NODE
%Keeping this in mind, consider the profile plot of the "core" node.
First, for the core-like node, whose profile plot is shown in 
Figure~\ref{subfig:tc}, the volume of the output cuts grows relatively 
smoothly as $\gamma$ is increased (\emph{i.e.}, as $-\gamma$ is decreased). 
For small $\gamma$, \emph{e.g.}, $\gamma=-0.0463$ or $\gamma=-0.0207$, the 
output cuts are forced to be small and hence display high conductance, as 
the region around the node is somewhat expander-like. 
By decreasing the teleportation, the conductance progressively decreases, 
as the rounding starts  % SAY BETTER ALGORITHM or probability mass spreads
to hit nodes in peripheral regions, whose inclusion only improves 
conductance (since it increases the cut volume without adding many 
additional cut edges).
In this case, this phenomena ends at $\gamma = -0.0013,$ when a cut 
of conductance value close to that of the global optimum is~found. 
(After that, larger and slightly better conductance cuts can still be found, 
but, as discussed below, they require $\gamma > 0$.)
\item
% INTERMEDIATE NODE
Second, a similar interpretation applies to the profile plot of the 
intermediate node, as shown in Figure~\ref{subfig:tb}. 
Here, however, the global component of the network containing the seed has 
smaller volume, around $300$, and a very low conductance (again, requiring 
$\gamma > 0$).
Thus, the profile plot \emph{jumps} from this cut to the much larger 
eigenvector sweep cut, as will be discussed below.
\item
% PERIPHERY NODE
Third, a more extreme case is that of the periphery-like node, whose profile 
plot is displayed in Figure~\ref{subfig:ta}. 
In this case, an initial increase in $\gamma$ does \emph{not} yield larger cuts.
This vertex is contained in a small-volume cut of low conductance, and thus
diffusion-based methods get ``stuck'' on the small side of the cut.
The only cuts of lower conductance in the network are those separating the 
global components, which can only be accessed when $\gamma> 0$. 
Hence, the teleportation must be greatly decreased before the algorithm 
starts outputting cuts at larger volumes.
(As an aside, this behavior is also often seen with so-called ``whiskers'' 
in much larger social and information 
networks~\cite{LLDM08_communities_CONF,LLDM09_communities_IM,LLM10_communities_CONF}.) 
\end{itemize}

\noindent
In addition, several general points that are illustrated in 
Figure~\ref{fig:teleportation} are worth emphasizing about the 
behavior of our algorithm.
\begin{itemize}
\item
First, \textsf{LocalCut} found low-conductance cuts of different volumes 
around each seed vertex, outperforming the shortest-path algorithm (as it 
should) by a factor of roughly $4$ in most cases. 
However, the results of \textsf{LocalCut} still lie away from the lower 
bound, which is also a factor of roughly $4$ smaller at most volumes. 
\item
% COMPARISON BETWEEN DIFFERENT SEED NOTES
Second, consider the range of the teleportation parameter necessary for the 
\textsf{LocalCut} algorithm to discover the well-balanced globally-optimal 
spectral partition.
In all three cases, it was necessary to make $\gamma$ positive 
(\emph{i.e.}, $-\gamma$ negative) to detect the well-balanced global 
spectral cut.
Importantly, however, the quantitative details depend strongly on whether the
seed is core-like, intermediate, or periphery-like.
That is, by \emph{formally} allowing ``negative teleportation'' 
probabilities, which correspond to $\gamma > 0$, the use of 
\emph{generalized} Personalized PageRank vectors as an exploratory tool is 
much stronger than the usual Personalized 
PageRank~\cite{andersen06local,andersen06seed}, in that it permits one to 
find a larger class of clusters, up to and including the global partition 
found by the solution to the global \textsf{Spectral} program.
Relatedly, it provides a smooth interpolation between Personalized PageRank
and the second eigenvector of the graph. 
Indeed, for $\gamma = 0.0028 \approx \lambda_2(G)$, \textsf{LocalCut} 
outputs the same cut as the eigenvector sweep cut for all three seeds.
\item
% MECHANISM
Third, recall that, given a teleportation parameter $\gamma$, the rounding 
step selects the cut of smallest conductance along the sweep cut of the 
solution vector.
(Alternatively, if volume-constrained sweeps are considered, then it selects 
the cut of smallest conductance among sweep cuts of volume at most 
$c \cdot k_\gamma$, where $k_\gamma$ is the lower bound obtained from the 
optimization program.)
In either case, increasing $\gamma$ can lead \textsf{LocalCut} 
to pick out larger cuts, but it does not \emph{guarantee} this will happen. 
In particular, due to the local topology of the graph, in many instances 
there may \emph{not} be a way of slightly increasing the volume of a cut 
while slightly decreasing its conductance. 
In those cases, \textsf{LocalCut} may output the same small sweep cut for a 
range of teleportation parameters until a much larger, much lower-conductance 
cut is then found.
The presence such horizontal and vertical \textit{jumps} in the local 
profile plot conveys useful information about the structure of the network 
in the neighborhood of the seed at different size scales, illustrating that 
the practice follows the theory quite well.
\end{itemize}

\subsubsection{Varying the Output-Size Parameter}
\label{sec:size}

Here, we evaluate the effect of varying the size factor $c$, for a fixed 
choice of teleportation parameter $\gamma$.
(In the previous section, $c$ was fixed at $c=2$ and $\gamma$ was varied.)
We have observed that varying $c$, like varying $\gamma$, tends to have the 
effect of producing low-conductance cuts of different volumes around the 
seed vertex. 
Moreover, it is possible to obtain low-conductance large-volume cuts, even 
at lower values of the teleportation parameter, by increasing $c$ to a 
sufficiently large value.
This is illustrated in Figure~\ref{fig:sizecore}, which shows the result of 
varying $c$ with the core-like node as the seed and $-\gamma = 0.02$. 
Figure~\ref{subfig:tc} illustrated that when $c=2,$ this setting only yielded 
a cut of volume close to $100$ (see the red triangle with $-\gamma=0.0207$);
but the yellow crosses in Figure~\ref{fig:sizecore} illustrate that by 
allowing larger values of $c$, better conductance cuts of larger volume can 
be~obtained. 

\begin{figure}[h]
%\begin{figure}[t]
\centering
\includegraphics[viewport= 50 -90 818 818, clip=true, scale=\imscale]{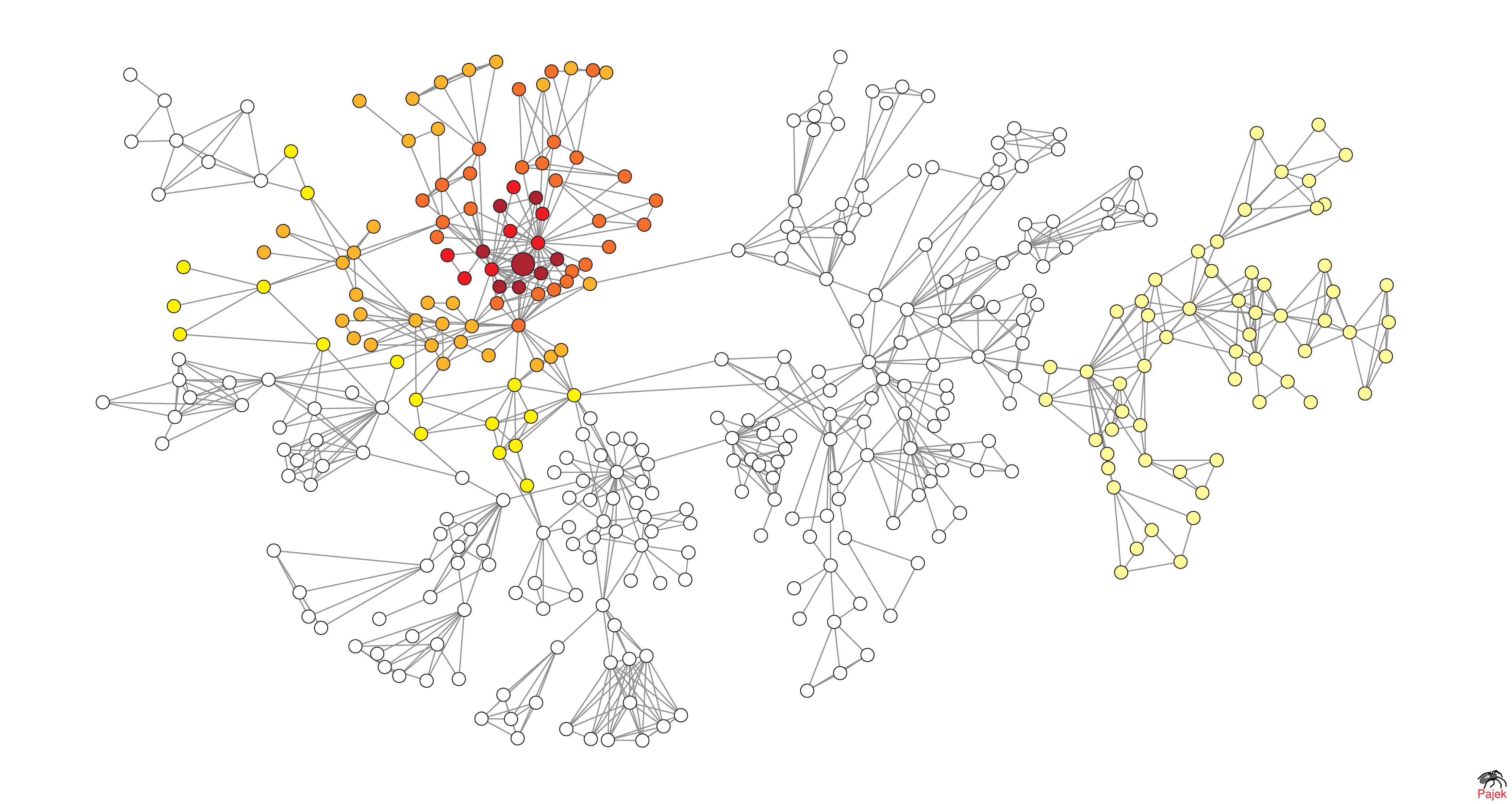}
\hspace{.7in}
\includegraphics[ scale=\plscale]{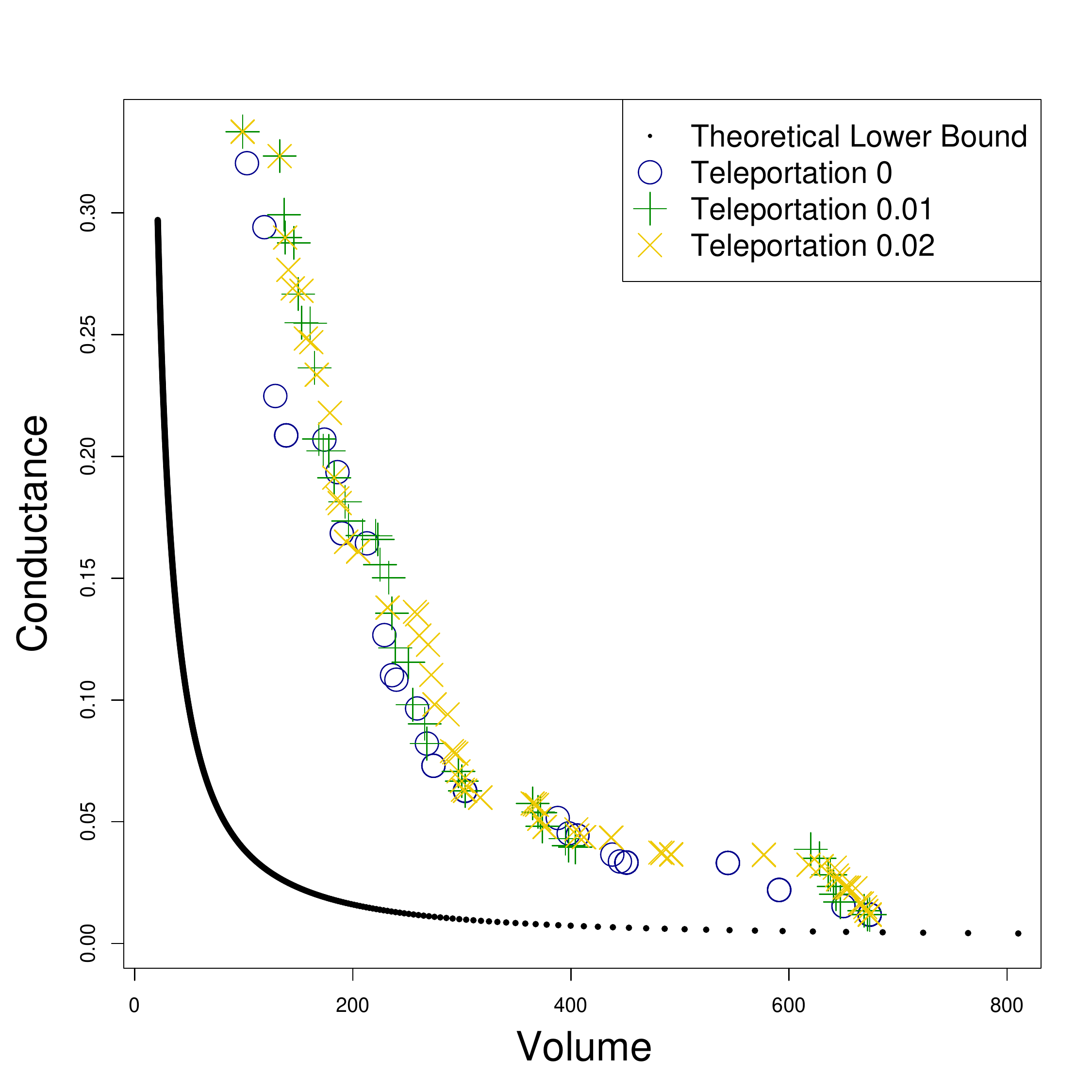}
\caption{[Best viewed in color.]  Selected cuts and local profile plots for varying $c$ with the core-like node as the seed. The cuts are displayed by assigning to each vertex a color corresponding to the smallest selected cut in which the vertex was included. Smaller cuts are darker, larger are lighter. The seed vertex is shown larger. The profile plot shows results from $1000$ runs of \textsf{LocalCut}, with varying $c$ and $-\gamma \in \{0, 0.01, 0.02\}$ . }
\label{fig:sizecore}
\end{figure}

While many of these cuts tend to have conductance slightly worse than the 
best found by varying the teleportation parameter, the observation that cuts 
of a wide range of volumes can be obtained with a single value of $\gamma$ 
leaves open the possibility that there exists a single choice of 
teleportation parameter $\gamma$ that produces good low-conductance cuts at 
all volumes simply by varying~$c$. 
(This would allow us to only solve a single optimization problem and still 
find cuts of different volumes.)
To address (and rule out) this possibility, we selected three choices of 
the teleportation parameter for each of the three seed nodes, and then we 
let $c$ vary. 
The resulting output cuts for the core-like node as the seed are plotted 
(in blue, green, and yellow) in Figure~\ref{fig:sizecore}. 
(The plots for the other seeds are similar and are not displayed.)
Clearly, no single teleportation setting dominates the others: in 
particular, at volume $200$ the lowest-conductance cut was produced with 
$-\gamma=0.02$; at volume $400$ it was produced with $-\gamma=0.01$; and at 
volume $600$ with it was produced with $\gamma=0$. 
The highest choice of $\gamma=0$ performed marginally better overall, 
recording lowest conductance cuts at both small and large volumes. 
That being said, the results of all three settings roughly track each~other, 
and cuts of a wide range of volumes were able to be obtained by varying the 
size parameter $c$.

These and other empirical results suggest that the best results are 
achieved when we vary both the teleportation parameter and the size factor. 
In addition, the use of multiple teleportation choices have the side-effect 
advantage of yielding multiple lower bounds at different volumes.

\subsubsection{Multiple Seeds and Correlation}
\label{sec:multi}

Here, we evaluate the behavior of \textsf{LocalCut} on more general seed 
vectors. 
We consider two examples---for the first example, there is an interpretation 
as a cut or partition consisting of multiple nodes; while the second example 
does not have any immediate interpretation in terms of cuts or partitions.

\begin{figure}[h] 
\subfigure[Seed set of four seed nodes.]{
\includegraphics[viewport= 50 0 818 818, clip=true, scale=\imscale]{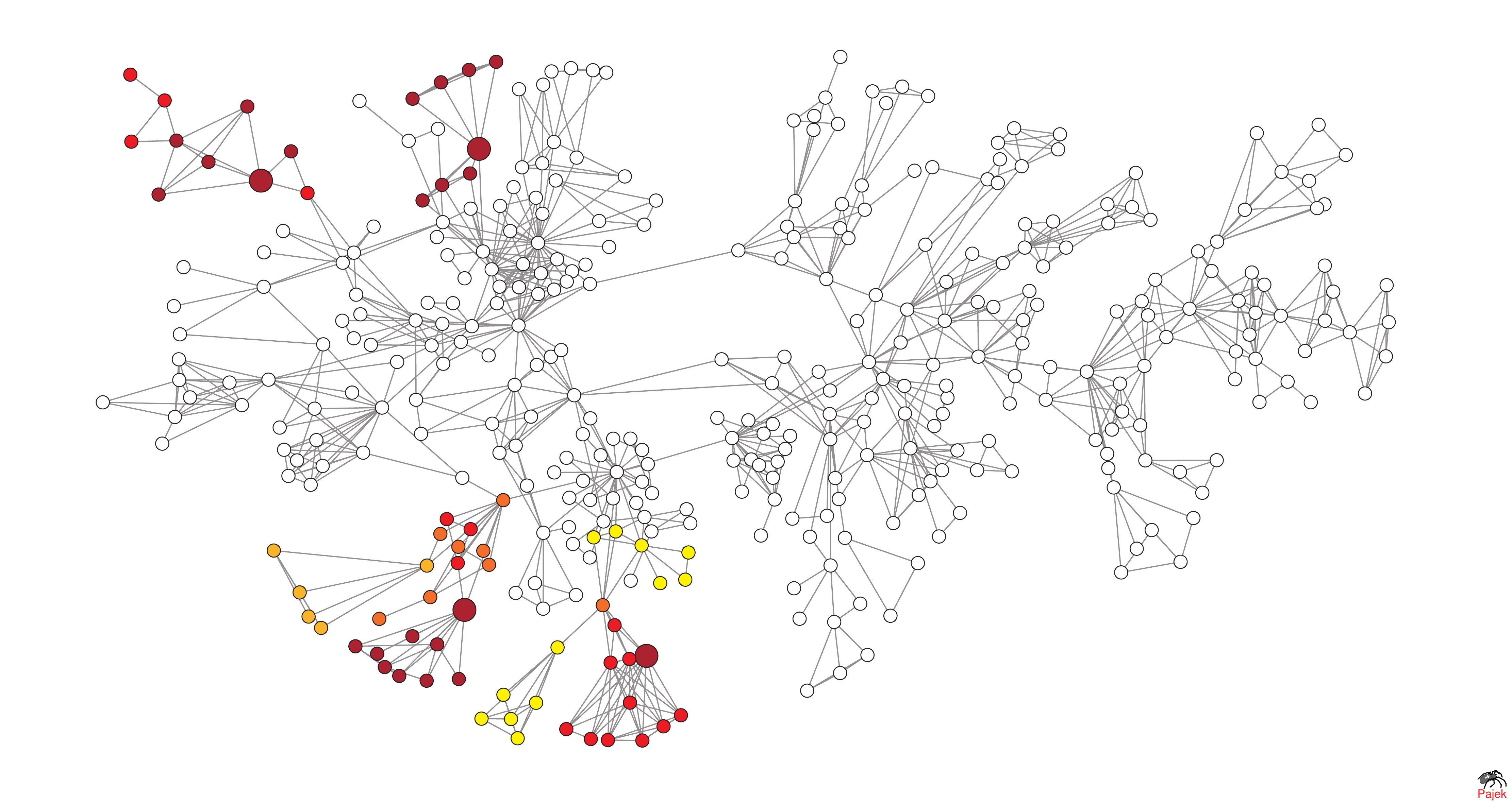}
\label{fig:mseeds0a}
}
\subfigure[A more general seed vector.]{
\includegraphics[viewport= 50 0 1250 818, clip=true, scale=\imscale]{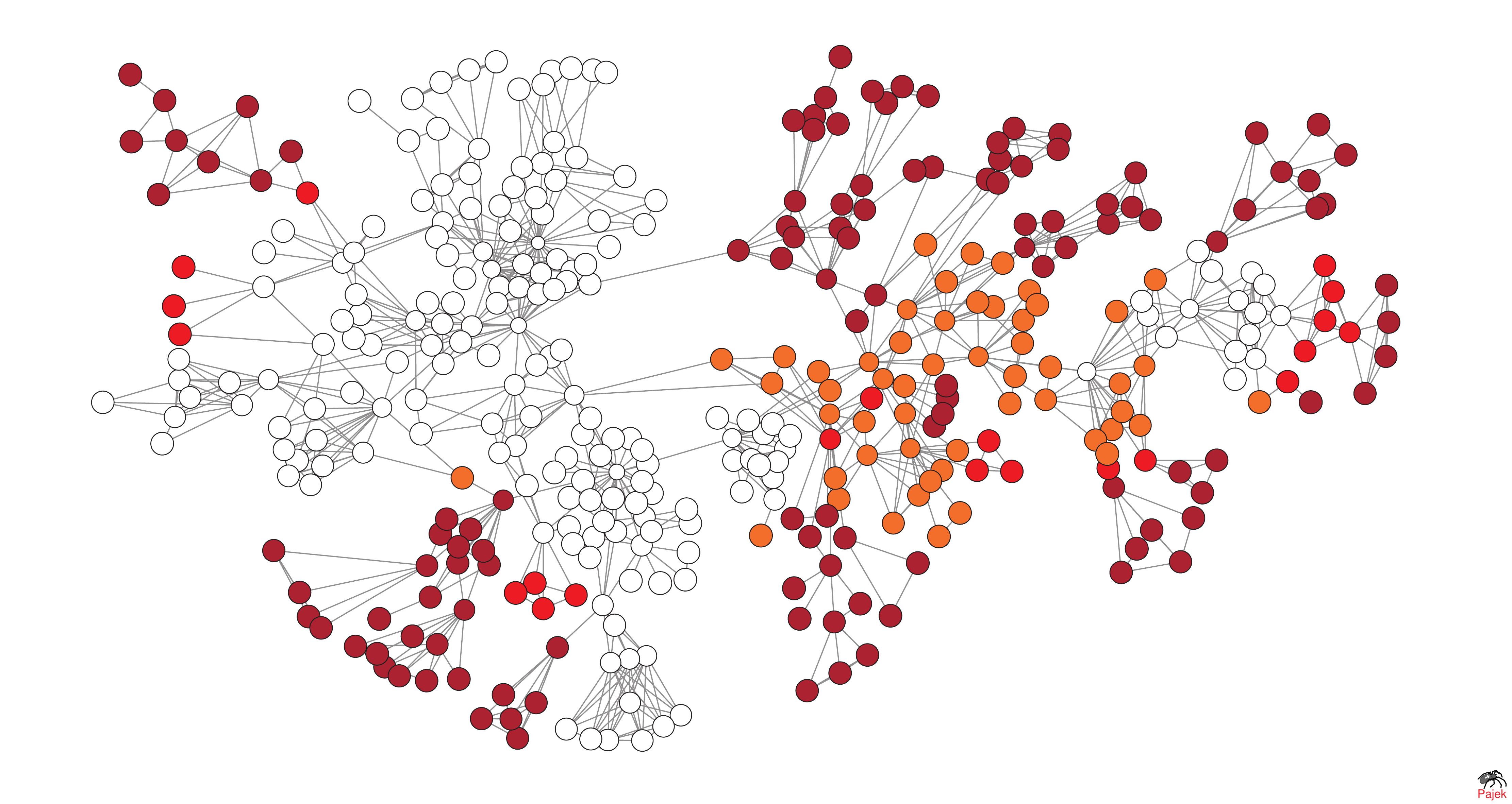}
\label{fig:mseeds0b}
}
\caption{[Best viewed in color.]  
Multiple seeds and correlation.
\ref{fig:mseeds0a} shows selected cuts for varying $\gamma$ with the seed 
vector corresponding to a subset of $4$ vertices lying in the periphery-like
region of the network. 
\ref{fig:mseeds0b} shows selected cuts for varying $\gamma$ with the seed 
vertex equal to a normalized version of the degree vector. 
In both cases, the cuts are displayed by assigning to each vertex a color 
corresponding to the smallest selected cut in which the vertex was included. 
Smaller cuts are darker, larger are lighter.  
%In \ref{fig:mseeds0a}, the seed vertices are shown larger, while in   
%\ref{fig:mseeds0b} the size of each vertex is an affine function of its 
%degree, where smaller degree corresponds to larger size.  
}
\label{fig:mseeds0}
\end{figure}

%% TWO FIGS %% \begin{figure}[h] 
%% TWO FIGS %% %\begin{figure}[t] 
%% TWO FIGS %% \centering
%% TWO FIGS %% \includegraphics[viewport= 50 0 818 818, clip=true, scale=\imscale]{image/mseedsperiphery.pdf}
%% TWO FIGS %% \caption{[Best viewed in color.]  Selected cuts for varying $\gamma$ with the seed vector corresponding to a subset of $4$ vertices lying in the periphery of the network. The cuts are displayed by assigning to each vertex a color corresponding to the smallest selected cut in which the vertex was included. Smaller cuts are darker, larger are lighter.  The seed vertices are shown larger.  }
%% TWO FIGS %% \label{fig:mseeds1}
%% TWO FIGS %% \end{figure}
%% TWO FIGS %% 
%% TWO FIGS %% 
%% TWO FIGS %% \begin{figure}[h]
%% TWO FIGS %% %\begin{figure}[t]
%% TWO FIGS %% \centering
%% TWO FIGS %% %\begin{center}
%% TWO FIGS %% \includegraphics[viewport= 50 0 1250 818, clip=true, scale=\imscale]{image/mseeds-degbias.pdf}
%% TWO FIGS %% \caption{[Best viewed in color.]  Selected cuts for varying $\gamma$ with the seed vertex equal to a normalized version of the degree vector. The cuts are displayed by assigning to each vertex a color corresponding to the smallest selected cut in which the vertex was included. Smaller cuts are darker, larger are lighter.  The size of each vertex is an affine function of its degree, where smaller degree corresponds to larger size.  }
%% TWO FIGS %% %\end{center}
%% TWO FIGS %% \label{fig:mseeds2}
%% TWO FIGS %% \end{figure}

In our first example, we consider a seed vector representing a subset of 
four nodes, located in different peripheral branches of the left half of 
the global partition of the the network: 
see the four slightly larger (and darker) vertices in Figure~\ref{fig:mseeds0a}.
This is of interest since, depending on the size-scale at which one is 
interested, such sets of nodes can be thought of as either ``nearby'' or ``far 
apart.''
For example, when viewing the entire graph of $379$ nodes, these 
four nodes are all close, in that they are all on the left side of the 
optimal global spectral partition; but when considering smaller clusters 
such as well-connected sets of $10$ or $15$ nodes, these four nodes are 
much farther apart. 
In Figure~\ref{fig:mseeds0a}, we display a selection of the cuts found by 
varying the teleportation, with $c=2$. 
The smaller cuts tend to contain the branches in which each seed node is 
found, while larger cuts start to incorporate nearby branches. 
Not shown in the color-coding is that the optimal global spectral 
partition is eventually recovered.
Identifying peripheral areas that are well-separated from the rest of the 
graph is a useful primitive in studying the structure of social 
networks~\cite{LLDM08_communities_CONF,LLDM09_communities_IM,LLM10_communities_CONF}; 
and thus, this shows how \textsf{LocalCut} may be used in this context, when 
some periphery-like seed nodes of the graph are known.

In our second example, we consider a seed vector that represents a feature 
vector on the vertices but that does not have an interpretation in terms of 
cuts. 
In particular, we consider a seed vector that is a normalized version of the 
degree distribution vector. 
Since nodes that are periphery-like tend to have lower degree than those that
are core-like~\cite{LLDM08_communities_CONF,LLDM09_communities_IM,LLM10_communities_CONF},
this choice of seed vector biases \textsf{LocalCut} 
towards cuts that are well-correlated with periphery-like and low-degree 
vertices.
A selection of the cuts found on this seed vector when varying the 
teleportation with $c=2$ is displayed in Figure~\ref{fig:mseeds0b}. 
These cuts partition the network naturally into three well-separated 
regions: a sparser periphery-like region in darker colors, a lighter-colored 
intermediate region, and a white dense core-like region, where higher-degree 
vertices tend to lie. 
Clearly, this approach could be applied more generally to find 
low-conductance cuts that are well-correlated with a known feature of the 
node vector.

\section{Discussion}
\label{sxn:discussion}

In this final section, we provide a brief discussion of our results in a 
broader context.

\paragraph{Relationship to local graph partitioning.} 
Recent theoretical work has focused on using spectral ideas to find good 
clusters nearby an input seed set of 
nodes~\cite{Spielman:2004,andersen06local,chung07_fourproofs}. 
In particular, local graph partitioning---roughly, the problem of finding a 
low-conductance cut in a graph in time depending only on the volume of the 
output cut---was introduced by Spielman and Teng~\cite{Spielman:2004}. 
They used random walk based methods; and they used this as a subroutine to 
give a nearly linear-time algorithm for outputting balanced cuts that match 
the Cheeger Inequality up to polylog factors. 
In our language, a local graph partitioning algorithm would start a random 
walk at a seed node, truncating the walk after a suitably chosen number of 
steps, and outputting the nodes visited by the walk. 
This result was improved  by Andersen, Chung and Lang~\cite{andersen06local} 
by performing a truncated Personalized PageRank computation. 
These and subsequent papers building on them were motivated by local graph 
partitioning~\cite{chung07_fourproofs}, but they do not address the problem 
of discovering cuts near general seed vectors, as do we, or of generalizing 
the second eigenvector of the Laplacian. 
Moreover, these approaches are more operationally-defined, while ours is 
axiomatic and optimization-based.

\paragraph{Relationship to empirical work on community structure.}
Recent empirical work has used Personalized PageRank, a particular variant 
of a local random walk, to characterize very finely the clustering and 
community structure in a wide range of very large social and information 
networks~\cite{andersen06seed,LLDM08_communities_CONF,LLDM09_communities_IM,LLM10_communities_CONF}. 
In particular, Andersen and Lang used local spectral methods to identify 
communities in certain informatics graphs using an input set of nodes as a 
seed set~\cite{andersen06seed}.
Subsequently, Leskovec, Lang, Dasgupta, and Mahoney used related 
methods to characterize the small-scale and large-scale clustering and 
community structure in a wide range of large social and information 
networks~\cite{LLDM08_communities_CONF,LLDM09_communities_IM,LLM10_communities_CONF}. 
Our optimization program and empirical results suggest that this line of 
work can be extended to ask in a theoretically principled manner much more 
refined questions about graph structure near prespecified seed vectors.

\paragraph{Relationship to cut-improvement algorithms.}
Many recently-popular algorithms for finding minimum-conductance cuts, 
such as those in~\cite{khandekar06_partitioning,OSVV08}, use as a crucial 
building block a primitive that takes as input a cut $(T, \bar{T})$ and 
attempts to find a lower-conductance cut that is {\em well correlated} 
with $(T, \bar{T})$. 
This primitive is referred to as a \emph{cut-improvement 
algorithm}~\cite{kevin04mqi,andersen08soda}, as its original purpose was 
limited to post-processing cuts output by other algorithms. 
Recently, cut-improvement algorithms have also been used to find low 
conductance cuts in specific regions of large graphs~\cite{LLM10_communities_CONF}. 
Given a notion of correlation between cuts, cut-improvement algorithms 
typically produce approximation guarantees of the following form:
for any cut $(C, \bar{C})$ that is $\varepsilon $-correlated with the input 
cut, the cut output by the algorithm has conductance upper-bounded by a 
function of the conductance of $(C, \bar{C})$ and $\varepsilon $.
This line of work has typically used flow-based techniques.
For example, Gallo, Grigoriadis and Tarjan~\cite{Gallo:1989} were the first 
to show that one can find a subset of an input set $T \subseteq V$ with 
minimum conductance in polynomial time. 
Similarly, Lang and Rao~\cite{kevin04mqi} implement a closely related 
algorithm and demonstrate its effectiveness at refining cuts output by other 
methods. 
Finally, Andersen and Lang~\cite{andersen08soda} give a more general 
algorithm that uses a small number of single-commodity maximum-flows to find 
low-conductance cuts not only inside the input subset $T$, but among all 
cuts which are well-correlated with $(T, \bar{T})$. 
Viewed from this perspective, our work may be seen as a spectral analogue 
of these flow-based techniques, since Theorem~\ref{thm:improve} provides 
lower bounds on the conductance of other cuts as a function of how 
well-correlated they are with the seed vector.

\paragraph{Alternate interpretation of our main optimization program.}
There are a few interesting ways to view our local optimization problem 
of Figure~\ref{fig:spectral} which would like to point out here. 
Recall that \textsf{LocalSpectral} may be interpreted as augmenting the 
standard spectral optimization program with a constraint that the output 
cut be well-correlated with the input seed set.
To understand this program from the perspective of the dual, recall that 
the dual of \textsf{LocalSpectral} is given by the following.
\begin{eqnarray*}
\label{prog:spectral-local-d1}
                     &\text{maximize} & \alpha + \beta \kappa                    \\
                     &\text{s.t.} & L_{G} \succeq \alpha L_{K_n} + \beta \Omega_T \\
                     &            & \beta \ge 0    ,
\end{eqnarray*}
where $\Omega_T=D_Gs_Ts_T^TD_G$.
Alternatively, by subtracting the second constraint of 
\textsf{LocalSpectral} 
from the first constraint, it follows that
$$
x^T\left(L_{K_n}-L_{K_n}s_Ts_T^TL_{K_n}\right)x \le 1-\kappa  .
$$
It can be shown that
$$
L_{K_n}-L_{K_n}s_Ts_T^TL_{K_n} 
   = \frac{L_{K_{T}}}{\vol(\bar{T})} + \frac{L_{K_{\bar{T}}}}{\vol(T)}  ,
$$
where $L_{K_{T}}$ is the $D_G$-weighted complete graph on the vertex set $T$.
Thus,
\textsf{LocalSpectral}
is clearly equivalent to
\begin{eqnarray*}
\label{prog:spectral-local-p2}
                              &\text{minimize} & x^T  L_{G} x                \\
                              &\text{s.t.} & x^T  L_{K_n} x = 1      \\
                              &            & x^T\left( \frac{L_{K_{T}}}{\vol(\bar{T})} + \frac{L_{K_{\bar{T}}}}{\vol(T)} \right)x \le 1-\kappa   .
\end{eqnarray*}
The dual of 
this program 
is given by the following.
\begin{eqnarray*}
\label{prog:spectral-local-d2A}
                             &\text{maximize} & \alpha - \beta(1-\kappa)   \\
\label{prog:spectral-local-d2B}
                             &\text{s.t.} & L_{G} \succeq \alpha L_{K_n} - \beta\left( \frac{L_{K_{T}}}{\vol(\bar{T})} + \frac{L_{K_{\bar{T}}}}{\vol(T)} \right)  \\
\label{prog:spectral-local-d2C}
                             &            & \beta \ge 0      . 
\end{eqnarray*}
From the perspective of this dual, this 
can be viewed as ``embedding'' a combination of a complete graph $K_n$ and a
weighted combination of complete graphs on the sets $T$ and $\bar{T}$, 
\emph{i.e.}, $K_T$ and $K_{\bar{T}}$.
Depending on the value of $\beta$, the latter terms clearly discourage cuts 
that substantially cut into $T$ or $\bar{T}$, thus encouraging partitions
that are well-correlated with the input cut $(T,\bar{T})$.

\paragraph{Bounding the size of the output cut.}
Readers familiar with the spectral method may recall that given a graph with a small balanced cut, it is not possible, in general, to guarantee that the sweep cut procedure of Theorem~\ref{thm:cheeger2}  applied to the optimal of \textsf{Spectral} outputs  a balanced cut. One may have to iterate several times before one gets a balanced cut.  Our setting, building up on the spectral method,  also suffers from this;  we cannot hope, in general,  to bound the size of the output cut  (which is a sweep cut) in terms of the correlation parameter $\kappa.$ 
This was the reason for considering volume-constrained sweep cuts in our
empirical evaluation.

%% %\newpage
%% %-----------------------------------------------------------------------
%% % \setlinespacing{1}
%% \bibliographystyle{plain}
%% %\bibliographystyle{unsrt}
%% %%%{\small
%% %\bibliography{nisheeth-local}
%% %%%}
%% \bibliography{communities,mwmbib_jrnl,mwmbib_proc,mwmbib_book,mwmbib_misc,mwmbib_drft,haribib,communities}
%% %-----------------------------------------------------------------------

\end{document}